\documentclass{amsart}

\subjclass{68Q17, 68Q25}

\usepackage{graphicx} 
\usepackage[utf8]{inputenc}
\usepackage{amsmath}
\usepackage{amssymb}
\usepackage{tikz}
\usepackage{amsthm}
\usepackage{hyperref}
\usepackage[margin=3.25cm]{geometry}
\usepackage{url}
\usepackage{mathabx}
\usepackage{theoremref}

\newtheorem{theorem}{Proposition}[section]
\newtheorem{corollary}[theorem]{Corollary}
\newtheorem{lemma}[theorem]{Lemma}
\newtheorem{main}[theorem]{Theorem}
\newtheorem{example}[theorem]{Example}

\title[Promise Systems of Equations]{Promise Systems of Equations over Magmas With Identity and over Algebras in Congruence Modular Varieties }
\author{Nick Jamesson}
\email{nicholas.jamesson@colorado.edu}
\address{Department of Mathematics,
University of Colorado Boulder, USA}

\thanks{Supported by the  National Science Foundation under Grant No. DMS 2452289}
\keywords{constraint satisfaction, promise constraint satisfaction, equations, minions, magma, congruence modular variety}

\date{\today}

\theoremstyle{definition}
\newtheorem{definition}[theorem]{Definition}

\newcommand{\SysPol}{\mathrm{SysPol}}
\newcommand{\PEqn}{\mathrm{PEqn}}
\newcommand{\CSP}{\mathrm{CSP}}
\newcommand{\PCSP}{\mathrm{PCSP}}

\newcommand{\Pol}{\mathrm{Pol}}

\newcommand{\AIP}{\mathrm{AIP}}
\newcommand{\BLP}{\mathrm{BLP}}
\newcommand{\ar}{\mathrm{ar}}

\newcommand{\SysTerm}{\mathrm{SysTerm}}
\newcommand{\Con}{\mathrm{Con }}

\newcommand{\A}{\mathbf{A}}
\newcommand{\B}{\mathbf{B}}
\newcommand{\C}{\mathbf{C}}

\begin{document}

\maketitle

\begin{abstract}
We study the computational complexity of solving promise systems of equations over finite algebras. In this line of research one considers two algebras $\A$ and $\B$ of the same signature with a homomorphism $\A\to \B$, and the problem is the following:
\begin{align*}
& \PEqn(\A, \B):\\
&\textsc{input:} \text{ A system of equations in the signature of }\mathbf{ A} \text{ and }\B.\\
&\textsc{problem:} \text{ Is there a solution to the system of equations in }\A\text{ or not even a solution in }\B?
\end{align*}

We generalize the results of Larrauri, Mottet, and Živný from \cite{larrauri2025equationsfinitemonoidsinfinite} to obtain a $\mathbf{P}-\mathbf{NP}$-hard dichotomy result for promise systems of equations over a class of algebras which contains all monoids, and a dichotomy result for promise systems of equations over algebras in a congruence modular variety. We then consider the metaproblem for promise systems of equations over algebras in a congruence modular variety: given finite algebras $\A$ and $\B$ such that $\A$ is in a congruence modular variety, we show there is a quasi-polynomial time algorithm for determining whether or not $\PEqn(\A, \B)$ is in $\mathbf{P}$.
\end{abstract}


\section{Introduction}
Constraint satisfaction problems (CSPs) have been an important and wide ranging topic of research since Schaefer's dichotomy result on Boolean domains \cite{Schaefer1}. For a CSP with a Boolean domain $\mathbb{A}$, Schaefer proved that either $\CSP(\mathbb{A})$ is in $\mathbf{P}$ or $\CSP(\mathbb{A})$ is $\mathbf{NP}$-complete. A similar dichotomy for CSPs with any finite domain has been shown to hold by Bulatov \cite{BulatovDichotomy} and by Zhuk \cite{ZhukDichotomy}.

Given the settled dichotomy of CSPs, the so called CSP metaproblem has also been studied in several papers (see \cite{Metaproblem1, Mayr23}). Here, the input is a relational structure $\mathbb{A}$, and the problem is to determine if $\CSP(\mathbb{A})$ is in $\mathbf{P}$ or is $\mathbf{NP}$-complete.

Brakensiek and Guruswami proposed a variant of CSPs known as promise constraint satisfaction problems (PCSPs) \cite{BrakensiekGuruswami1}. A natural example of a PCSP is a graph coloring problem: is an input graph 3-colorable or not even 6-colorable. The \emph{promise} in the PCSP is that every input graph is indeed either 3-colorable or not even 6-colorable. Computational complexity of general PCSPs is generally still a wide open problem.

In this paper, we study promise systems of equations over finite algebras. This is a variant of a PCSP in which we fix algebras $\A$ and $\B$ with a homomorphism from $\A$ to $\B$ and have the following computational problem:

\begin{align*}
& \PEqn(\A, \B)\colon\\
&\textsc{input:} \text{ A system of equations in the signature of }\mathbf{ A} \text{ and }\B.\\
&\textsc{problem:} \text{ Is there a solution to the system of equations in }\A\text{ or not even a solution in }\B?
\end{align*}
The study of $\PEqn(\A, \B)$ was initiated in \cite{larrauri2024solving} by Larrauri and Živný in the case of finite semigroups.

Our main results are dichotomy theorems for the computational complexity of $\PEqn(\A, \B)$: we characterize when $\PEqn(\A, \B)$ is in $\mathbf{P}$ and show that $\PEqn(\A, \B)$ is $\mathbf{NP}$-hard otherwise. In particular, we establish dichotomy theorems for a class of algebras which contains expansions of magmas (see Definition \ref{Magma}) with a two sided identity, and also for algebras in a congruence modular variety. As a consequence, we obtain dichotomy results for familiar algebras contained in these classes such as lattices and rings. In order to formally state these results, we will first need some introductory definitions. We refer the reader to \cite{BurrisSanka} for background in universal algebra. Throughout, we will be working with relational structures and algebras. We will denote by $\mathbb{N}$ the set of natural numbers $\lbrace 0 , 1, 2, \dots \rbrace$. We will also use the following notation for $n\in \mathbb{N}$ often: $[n] = \lbrace   1, 2, \dots , n \rbrace.$
 
 \begin{definition}
 \thlabel{signature}
     A \emph{signature} is a set of symbols; either function symbols, or relation symbols, each with an associated natural number, which we call the \emph{arity} of the symbol (0-ary functions symbols are just constant symbols). If $f$ is a function symbol, we denote its arity by $\ar(f)$ and similarly for a relation symbol $R$. 
 \end{definition}

 \begin{definition}
 \thlabel{RelationalStructure}
 Let $\sigma$ be a signature with no function symbols. A \emph{relational structure} $\mathbb{A}$ of signature $\sigma$ is a set $A$ together with a relation $R^{\mathbb{A}}\subseteq A^{\ar(R)}$ for each relation symbol $R\in \sigma$. In this case, we write $\mathbb{A} = (A;  R_1^{\mathbb{A}}, R_2^{\mathbb{A}}\dots)$. We say $\mathbb{A}$ is \emph{finite} if $A$ is finite and $\sigma$ is finite.
 \end{definition}
 
\begin{definition}
\thlabel{Algebra}
    Let $\sigma$ is a signature with no relation symbols. An \emph{algebra} $\A$ of signature $\sigma$ is a set $A$ together with a function $f^{\A}\colon A^{\ar(f)}\to A$ for each function symbol $f\in \sigma$. We write $\A = (A; f_1^{\A}, f_2^{\A}, \dots)$. We say $\A$ is \emph{finite} if $A$ is finite and $\sigma$ is finite.
\end{definition}

Note that if $\ar(f) = 0$, then $f^\A\in A$ is just constant.

\begin{definition}
\thlabel{homomorphismRelation}
    Let $\mathbb{A}$ and $\mathbb{B}$ be two finite relational structures of the same signature $\sigma$. A function $\phi \colon \mathbb{A}\to \mathbb{B}$ is called a \emph{homomorphism} if the following hold: for all $R\in \sigma$ and for all $(a_1, \dots , a_{\ar({R})})\in R^{\mathbb{A}}$, we have $(\phi(a_1), \dots , \phi(a_{\ar({R})}))\in R^{\mathbb{B}}$.
\end{definition}

\begin{definition}
\thlabel{homomorphismAlgebra}
    Let $\A$ and $\B$ be algebras of the same signature $\sigma$. A function $\phi \colon \A\to \B$ is called a \emph{homomorphism} if for all $f\in \sigma$, and for all $(a_1, \dots , a_{\ar(f)})\in A^{\ar(f)}$, we have $\phi(f^{\A}(a_1, \dots , a_{\ar(f)})) = f^{\B}(\phi(a_1), \dots , \phi(a_{\ar(f)}))$.
\end{definition}

We are now ready to give the definition of a \emph{promise constraint satisfaction problem}.

\begin{definition}
\thlabel{PCSPdef}
\cite{BrakensiekGuruswami1}
Let $\mathbb{A}$ and $\mathbb{B}$ be two relational $\sigma$-structures with a homomorphism $\mathbb{A}\to \mathbb{B}$. Then $\PCSP(\mathbb{A},\mathbb{B})$ is the computational problem given as follows: 

\textsc{input}: A finite relational $\sigma$-structure $\mathbb{I}$.

\textsc{output}: \textsc{yes} if there is a homomorphism $\mathbb{I}\to\mathbb{A}$. \textsc{no} if there is no homomorphism $\mathbb{I}\to \mathbb{B}$.
\end{definition}

The \emph{promise} here is that every input instance $\mathbb{I}$ is indeed either a \textsc{yes} instance or a \textsc{no} instance. The existence of a homomorphism $\mathbb{A}\to \mathbb{B}$ guarantees that the input instance $\mathbb{I}$ is not both a \textsc{yes} instance and a  \textsc{no} instance.

The main topic of this paper is a variant of a PCSP, namely a \emph{promise system equations} over algebras, which we will define next. A \emph{system of equations} in a signature $\sigma$ without relation symbols is a finite set of equations of the form $s(x_1, \dots , x_n) = t(x_1, \dots , x_n)$ where $s$ and $t$ are $\sigma$-terms. If the variables occurring in the system of equations are $x_1, x_2, \dots , x_n$, we may denote the system $\Sigma(x_1, x_2, \dots ,x_n)$. If $\A$ is an algebra of signature $\sigma$, and $\Sigma(x_1, x_2, \dots ,x_n)$ is a system of equations in the signature $\sigma$, then a \emph{solution} to $\Sigma(x_1, x_2, \dots , x_n)$ over $\A$ is a mapping $h\colon \lbrace x_1, x_2, \dots , x_n \rbrace \to A$ such that for each equation $$s(x_1, \dots , x_n) = t(x_1, \dots , x_n)\in \Sigma(x_1, x_2, \dots , x_n)$$ we have that $$s^{\A}(h(x_1), \dots , h(x_n)) = t^{\A}(h(x_1), \dots , h(x_n)).$$

\begin{definition}
\thlabel{PromiseEqsDef}
    Let $\A$ and $\B$ be finite algebras of signature $\sigma$ with a homomorphism $\A\to \B$. We define the \emph{promise system of equations} $\PEqn(\textbf{A},\textbf{B})$ to be the following computational problem:
    
    \textsc{input:} A system of term equations in the signature $\sigma$.
    
    \textsc{output:} \textsc{yes} if there is a solution to the system in $\A$. \textsc{no} if there is no solution to the system in $\B$.
\end{definition}
As in the PCSP case, the promise here is that our instance of the problem is either a \textsc{yes} instance or a \textsc{no} instance. We denote $\text{SysTerm}(\A) = \PEqn(\A, \A)$. Given an algebra $\A$ of signature $\sigma$, we may expand $\sigma$ to include a constant symbol for each $a\in A$ and obtain an algebra which interprets the new constant symbols in the natural way. If we denote the resulting algebra $\A'$, then a system of term equations over $\A'$ is just a system of polynomial equations over $\A$. We denote $\SysPol(\A) = \text{SysTerm}(\mathbf{A'})$. There is a trivial reduction from $\PEqn(\A, \B)$ to both $\SysTerm(\A)$ and to $\SysTerm(\B)$. In view of the CSP dichotomy theorem, this means the interesting cases of $\PEqn(\A, \B)$ are the cases in which both $\SysTerm(\A)$ and $\SysTerm(\B)$ are $\mathbf{NP}$-complete. The next example shows that $\PEqn(\A, \B)$ could be in $\mathbf{P}$ in such a case. \begin{example}
    \label{ExampleRings}
    Let $\A = (2\mathbb{Z}_8; +^{\mathbb{Z}_8},\cdot^{\mathbb{Z}_8} , 2^{\mathbb{Z}_8})$ and $\B = (\mathbb{Z}_4; +^{\mathbb{Z}_4},\cdot^{{\mathbb{Z}_4}}, 2^{\mathbb{Z}_4})$. Then $\SysTerm(\A)$ and $\SysTerm(\B)$ are $\mathbf{NP}$-complete, but $\PEqn(\A,\B)$ is in $\mathbf{P}$.
\end{example} Example \ref{ExampleRings} will follow from Theorem \ref{MainDichotomy}.

In \cite[Proposition 5.9, Proposition 6.2]{larrauri2025equationsfinitemonoidsinfinite}, Larrauri, Mottet, and Živný gave a dichotomy result for promise systems of equations over monoids expanded with an additional relation. The first of our main results generalizes this to expansions of magmas with an identity. In particular, we obtain a dichotomy theorem for $\PEqn(\A , \B)$ in the case where there is a binary polynomial operation with identity on $\A$ which need not be associative.

\begin{main}
    Let $\A$ and $\B$ be finite algebras, let $e\in A$, let $\cdot$ be a binary polynomial operation on $\A$ such that $x \cdot e = x= e \cdot x$ for all $x\in A$. Suppose there is a homomorphism $\A\to \B$. Then $\PEqn(\A, \B)$ is either in $\mathbf{P}$ or is $\mathbf{NP}$-hard.
\end{main}
We will prove this as Theorem $\ref{MainDichotomy}$ and characterize the conditions under which $\PEqn(\A , \B)$ is in $\mathbf{P}$. We then use Theorem \ref{MainDichotomy},  to prove Theorem \ref{DichotomyCongruenceModular}. This is a dichotomy for $\PEqn(\A, \B)$ in which $\A$ belongs to a congruence modular variety.

\begin{main}[Proven as Theorem \ref{DichotomyCongruenceModular}]
    Let $\A$ and $\B$ be finite algebras with a homomorphism $\A\to \B$ such that $\A$ is in a congruence modular variety. Then $\PEqn(\A , \B)$ is either in $\mathbf{P}$ or is $\mathbf{NP}$-hard.
\end{main}

Following these results, we prove Theorem \ref{MetaProblemMal'cev}, which shows there is a quasi-polynomial time algorithm that decides the metaproblem for promise systems of equations over finite algebras in a congruence modular variety.

\begin{main}\label{MetaProblemMal'cev}
    Let $\A$ and $\B$ be finite algebras with a homomorphism $\A\to \B$ such that $\A$ is in a congruence modular variety. Then there is a quasi-polynomial time algorithm which decides if $\PEqn(\A,\B)$ is in $\mathbf{P}$ or is instead $\mathbf{NP}$-hard.
\end{main}

\section{Equivalence of PCSPs and promise systems of equations}

We can reformulate a promise system of equations problem as a PCSP with relational structures whose relations are the graphs of the basic operations in our algebras. This allows us to use results on PCSPs to work on promise systems of equations.

\begin{definition}
    Let $A$ be a set and let $f\colon A^n\to A$ be a function. Then we define the relation $$f^{\circ} = \lbrace (x_1, \dots , x_n,y)\in A^{n+1}:f(x_1, \dots , x_n) = y \rbrace .$$
\end{definition}

Note that if $f$ is a 0-ary (constant) function, then $f^{\circ}$ is a singleton set.

\begin{definition}
    Let $\A$ be an algebra of signature $\lbrace f_1, \dots , f_n \rbrace$. We denote by $\mathbb{A}^\circ$ the relational structure $( A; (f_1^{\A})^\circ, \dots , (f_n^{\A})^\circ )$.
\end{definition}

In \cite[Theorem 2.2]{ZadoriLarose06}, Larose and Z\'{a}dori show that systems of polynomial equations over finite algebras are logspace equivalent to CSPs over finite relational structures. A similar proof works for the following:

\begin{theorem}
    \thlabel{CSP equiv to eq}
    Let $\A$ be a finite algebra of signature $\lbrace f_1, \dots , f_n \rbrace$. Then $\text{SysTerm}(\A)$ is logspace equivalent to $\CSP(\mathbb{A}^{\circ})$.
\end{theorem}

The proof of the above relies only on converting an instance of a system of term equations in a signature $\lbrace f_1, \dots , f_n \rbrace$ to an instance of a constraint satisfaction problem in a signature $\lbrace  f_1^{\circ}, \dots , f_n^{\circ}\rbrace$ and vice versa. The reductions do not depend on the structure $\A$ and can be recycled to prove the following equivalence:

\begin{theorem}
    \label{PEqnasPCSP}
    Let $\A$ and $\B$ be finite algebras of signature $\lbrace f_1, \dots , f_n \rbrace$ with a homomorphism $\A\to \B$. Then $\PEqn(\A, \B)$ is logspace equivalent to $\PCSP(\mathbb{A}^{\circ}, \mathbb{B}^{\circ})$.
\end{theorem}

By Proposition \ref{PEqnasPCSP}, we can use results about PCSPs in our discussion of promise equations. 

\section{Minions and polymorphisms}

\begin{definition}
\thlabel{polymorphismRelation}
Let $\mathbb{A}$, $\mathbb{B}$ be structures and let $n\in \mathbb{N}\backslash \lbrace 0 \rbrace$. We say $p\colon\mathbb{A}^n\to\mathbb{B}$ is a \emph{polymorphism} from $\mathbb{A}$ into $\mathbb{B}$ if $p$ is a homomorphism from $\mathbb{A}^n$ into $\mathbb{B}$. We denote the set of polymorphisms from $\mathbb{A}$ into $\mathbb{B}$ by $\Pol(\mathbb{A}, \mathbb{B})$, and we denote $\Pol(\mathbb{A}) = \Pol(\mathbb{A}, \mathbb{A})$.
\end{definition}

It is known that $\Pol(\mathbb{A}, \mathbb{B})$ determines the complexity of $\PCSP(\mathbb{A}, \mathbb{B})$ (see \cite[Theorem 2.25, Theorem 3.1]{AlgApproachToPSCP}). In the study of PCSPs, we work with polymorphisms between two different structures in $\Pol(\mathbb{A}, \mathbb{B})$, and therefore polymorphisms can no longer be composed as they can in $\Pol(\mathbb{A})$. However, $\Pol(\mathbb{A}, \mathbb{B})$ is still closed under variable manipulations, such as identification, permutation, and introduction of fictitious variables. Formally, $\Pol(\mathbb{A}, \mathbb{B})$ is a \emph{minion}.

\begin{definition}
    \label{AbstractMinion}
    \cite{larrauri2024solving}
    A \emph{minion} $\mathcal{M}$ consists of a set $\mathcal{M}^{(n)}$ for each $n\in \mathbb{N}\backslash \lbrace 0 \rbrace$, and a function $\pi^{\mathcal{M}}\colon\mathcal{M}^{(n)}\to \mathcal M^{(m)}$ for each function $\pi\colon [n]\to [m]$ satisfying:

    \begin{enumerate}
        \item[(1)] $id_{[n]}^\mathcal{M} = id_{\mathcal{M}^{(n)}}$ for all $n\in \mathbb{N}$.
        \item[(2)] $(\pi\circ \tau)^\mathcal{M} = \pi^{\mathcal{M}}\circ \tau^{\mathcal{M}}$ for all $\pi, \tau$ that can be composed.
    \end{enumerate}
    
\end{definition}

\begin{definition}
    \label{Minor}
    \cite[Definition 2.19]{AlgApproachToPSCP}
    Let $A$ and $B$ be sets, let $m,n\in \mathbb{N}$, let $f\colon A^n\to B$ be a function, and let $\pi\colon [n]\to [m]$ be a function. We define $f^{(\pi)}\colon A^m \to B$ by $$f^{(\pi)}(x_1, \dots , \dots x_m) = f(x_{\pi(1)}, \dots x_{\pi(n)})$$ for all $x_1, \dots x_m \in A$. We call $f^{(\pi)}$ a \emph{minor} of $f$. We denote by $f_\Delta$ the function $$f_\Delta\colon A\to B , \quad x\mapsto f(x, \dots , x)$$ and we call $f_\Delta$ the \emph{unary minor} of $f$.
\end{definition}

\begin{definition}
    \label{Minion}
    \cite[Definition 2.20]{AlgApproachToPSCP}
    Let $A$ and $B$ be sets and let $\mathcal{O}(A,B) = \lbrace f\colon A^n\to B : n\geq 1 \rbrace$. A \emph{function minion} $\mathcal{M}$ from $A$ to $B$ is a subset of $\mathcal{O}(A,B)$ closed under taking minors of functions. The set of $n$-ary elements of $\mathcal{M}$ is denoted $\mathcal{M}^{(n)}$.
\end{definition}

Note that if $\mathbb{A}$ and $\mathbb{B}$ are structures, then $\Pol(\mathbb{A}, \mathbb{B})$ is a function minion. A function minion is an example of a minion. 

In the proof of Theorem \ref{MainDichotomy}, we will make use of \emph{monoidal minions}.

\begin{definition}
    \label{CommutativeTuple}
    \cite{larrauri2024solving}
    Let $\mathbf{M}$ be a monoid and $(a_1, \dots , a_n)\in M^n$. We say $(a_1, \dots , a_n)$ is a \emph{commutative tuple} from $M$ if $a_ia_j = a_ja_i$ for all $i,j\in [n]$. 
\end{definition}

\begin{definition}
    \label{MonoidalMinion}
    \cite[Definition 3]{larrauri2024solving}
    Let $\mathbf{M}$ be a monoid and let $a\in M.$ The \emph{monoidal minion} $\mathcal{M}_{\mathbf{M},a}$ consists of commutative tuples from $M$ whose product is $a$. Given $(a_1, \dots , a_n)\in \mathcal{M}_{M,a}$, and $\pi\colon [n]\to [m]$, the minor $(a_1, \dots , a_n)^{(\pi)}$ is the tuple $(b_1, \dots , b_m)$ where $b_i = \prod_{j\in \pi^{-1}(i)}a_j$ for $i\in [m]$. Here, the empty product is just the identity $e$ of the monoid $\mathbf{M}$. Hence, if $\pi^{-1}(i)$ is empty, we get $b_i = \prod_{j\in \pi^{-1}(i)}a_j = e$ above.
\end{definition}

Note that monoidal minions are minions.

\begin{definition}
    \label{MinionHom}
    \cite[Definition 2.21]{AlgApproachToPSCP}
    Let $\mathcal{M}$ and $\mathcal{N}$ be minions. A function $\xi\colon\mathcal{M}\to \mathcal{N}$ is called a \emph{minion homomorphism} if the following hold:

    \begin{enumerate}
        \item[(1)] $\xi$ preserves arities. That is, for all $n\in \mathbb{N}\backslash \lbrace 0 \rbrace$ we have $\xi(\mathcal{M}^{(n)}) \subseteq \mathcal{N}^{(n)}$.
        \item[(2)] For all $m,n\in \mathbb{N}\backslash \lbrace 0 \rbrace$ and $\pi\colon[n]\to [m]$ and $f\in \mathcal{M}^{(n)}$, we have $\xi(f^{(\pi)}) = (\xi(f))^{(\pi)}$.
    \end{enumerate}

    If there are minion homomorphisms $\xi_1\colon \mathcal{M}\to \mathcal{N}$ and $\xi_2\colon \mathcal{N}\to \mathcal{M}$, we say $\mathcal{M}$ and $\mathcal{N}$ are \emph{homomorphically equivalent}.
\end{definition}

In \cite{BLP+AIP} and \cite{AlgApproachToPSCP}, the basic linear programming relaxation (BLP) and the basic affine integer relaxation (AIP) of an instance of a PCSP are studied. Given an instance $I$ of $\PCSP(\mathbb{A},\mathbb{B})$, one can construct an instance $\text{BLP}_\mathbb{A}(I)$ of a linear program corresponding to $I$. One can also construct the so called affine relaxation instance $\text{AIP}_\mathbb{A}(I)$. This is a system of equations over the integers. Under some conditions, algorithms which use these new instances will correctly solve $\PCSP(\mathbb{A},\mathbb{B})$. The AIP algorithm uses the instance $\text{AIP}_\mathbb{A}(I)$ to solve $\PCSP(\mathbb{A}, \mathbb{B})$. The BLP+AIP algorithm described in \cite{BLP+AIP} uses both instances. Both algorithms run in polynomial time.

\begin{definition}
    \cite[Definition 3.1]{BLP+AIP}
    The BLP+AIP algorithm \emph{correctly solves} $\PCSP(\mathbb{A}, \mathbb{B})$ if it accepts any \textsc{yes} instance of $\PCSP(\mathbb{A},\mathbb{B})$ and rejects any \textsc{no} instance of $\PCSP(\mathbb{A}, \mathbb{B})$.
\end{definition}

We will now describe the conditions required for BLP+AIP and for AIP to correctly solve $\PCSP(\mathbb{A},\mathbb{B})$. In particular, the polymorphism minion $\Pol(\mathbb{A},\mathbb{B})$ is required to contain functions which have some symmetry properties.

\begin{definition}
    \label{2BlockSymmetric}
    Let $A$ and $B$ be sets, let $n\in \mathbb{N}$, and let $f\colon A^{2n+1} \to B$ be a function. Then $f$ is called \emph{2-block symmetric} if there exists a partition of $[2n+1]$ into two blocks of size $n$ and $n+1$, respectively, such that for every permutation $\pi$ of $[2n+1]$ that preserves each block, we have $f = f^{(\pi)}$.
\end{definition}

\begin{definition}
    \label{alternating function}
    Let $A,B$ be sets, let $n\in \mathbb{N}$, let $f\colon A^{2n+1}\to B$ be a function. Then $f$ is called \emph{alternating} if for all $a_1, \dots a_{2n+1},x,y \in A$, we have $f(a_1, \dots ,a_{2n-1}, x,x) = f(a_1, \dots ,a_{2n-1}, y,y) $ and $f(a_1, \dots ,a_{2n+1}) = f(a_{\pi(1)}, \dots , a_{\pi(2n+1)})$ for all permutations $\pi$ on $[2n + 1]$ which send even numbers to even numbers and odd numbers to odd numbers.
\end{definition}

Note that an alternating function is 2-block symmetric with blocks corresponding to even and odd indices.

\begin{example}
    Let $(A; +)$ be an abelian group. Then $f\colon A^{2n+1}\to A$ given by $$f(x_1, \dots x_{2n+1}) = x_1 -x_2 +x_3 - \dots  -x_{2n}+x_{2n+1}$$ is an alternating function.
\end{example}

\begin{theorem}
    \label{BLP+AIP characterization}
    \cite[Theorem 5.1]{BLP+AIP}
    Let $\mathbb{A}$ and $\mathbb{B}$ be finite relational structures with a homomorphism $\mathbb{A}\to \mathbb{B}$. Then $\PCSP(\mathbb{A}, \mathbb{B})$ is solvable by $\BLP + \AIP$ if and only if $\Pol(\mathbb{A}, \mathbb{B})$ contains 2-block symmetric maps of all odd arities.
\end{theorem}

\begin{theorem}
    \label{AIP characterization}
    \cite[Theorem 7.19]{AlgApproachToPSCP}
    Let $\mathbb{A}$ and $\mathbb{B}$ be finite relational structures with a homomorphism $\mathbb{A}\to \mathbb{B}$. Then $\PCSP(\mathbb{A}, \mathbb{B})$ is solvable by $\AIP$ if and only if $\Pol(\mathbb{A}, \mathbb{B})$ contains alternating maps of all odd arities.
\end{theorem}

\section{Minor conditions and hardness from minions}

$\Pol(\A, \B)$ controls the complexity of $\PEqn(\A, \B)$ by \cite[Theorem 2.25, Theorem 3.1]{AlgApproachToPSCP}. The main tool for proving hardness of $\PEqn(\A , \B)$ that we will use is the existence of a \emph{selection function} on $\Pol(\A, \B)$.

\begin{definition}
    \label{SelectionFunction}
    Let $\mathcal{M}$ be a subset of a minion. Suppose there exist $K\in \mathbb{N}$ and $I\colon \mathcal{M}\to \mathcal{P}(\mathbb{N})$ which for all $m,n \in \mathbb{N}\backslash \lbrace 0 \rbrace$ satisfy the following properties:

    \begin{enumerate}
        \item[(1)] For each $p\in \mathcal{M}^{(n)}$, we have $I(p)\subseteq [n]$ and $|I(p)|< K$.
        \item[(2)] For each $\pi\colon [n]\to [m]$ and $p\in \mathcal{M}^{(n)}$ such that $p^{(\pi)}\in \mathcal{M}$, we have $ \pi(I(p)) \cap I(p^{(\pi)}) \neq \emptyset.$
    \end{enumerate}
Then the function $I$ is called a \emph{selection function} for $\mathcal{M}$ with bound $K$.
\end{definition}

In order to obtain a generalization of \cite[Corollary 5.13]{AlgApproachToPSCP}, which allows us to show that $\PCSP(\mathbb{A}, \mathbb{B})$ is $\mathbf{NP}$-hard by constructing selection functions, it is necessary to first define minor conditions. 

\begin{definition}
    Let $f$ and $g$ be function symbols of arities $m$ and $n$ respectively. A \emph{minor identity} is an expression of the form $$f(x_1, \dots , x_m ) = g(x_{\pi(1)}, \dots , x_{\pi(n)})$$ where $\pi\colon [n]\to [m]$. A \emph{minor condition} is a finite system of minor identities. A \emph{bipartite minor condition} is a minor condition in which the function symbols on the right hand side of the identities are disjoint from the function symbols on the left hand side of the identities.
\end{definition}

\begin{definition}
    Let $\mathcal{M}$ be a set of functions and let $\Sigma$ be a minor condition. We say that $\mathcal{M}$ satisfies $\Sigma$ if there is an assignment $\zeta$ of function symbols in $\Sigma$ to functions in $\mathcal{M}$ which preserves arity and such that for every identity $$f(x_1, \dots , x_n ) = g(x_{\pi(1)}, \dots , x_{\pi(m)}) \in \Sigma$$ we have $\zeta(f) = (\zeta(g))^{(\pi )}$. We say a minor condition is \emph{trivial} if it is satisfied in every function minion. Equivalently, a minor condition is trivial if it is satisfied in $\mathcal{P}_2$, the projections on the set $\lbrace 0 , 1 \rbrace$.
\end{definition}

\begin{definition}
    Let $\epsilon>0$. A bipartite minor condition $\Sigma$ is said to be $\epsilon$-\emph{robust} if no $\epsilon$-fraction of identities from $\Sigma$ is trivial.
\end{definition}

For functions $f\colon \mathbb{N}\to \mathbb{N}$ and $g\colon \mathbb{N}\to \mathbb{N}$, if there are constants $c$ and $M$ such that for all $n\geq M$ we have $f(n) \geq cg(n)$, then we say $f(n)\in \Omega(g(n))$.

\begin{theorem}
    \cite[Theorem 5.22]{AlgApproachToPSCP}
    \label{MinionUnion1}
    Let $\mathbb{A}$ and $\mathbb{B}$ be relational structures with a homomorphism $\mathbb{A}\to \mathbb{B}$, let $\epsilon\colon \mathbb{N}\to \mathbb{N}$ be a function with $\epsilon(l)\in \Omega(l^{-K})$ for all $K>0$, let $$\Pol(\mathbb{A}, \mathbb{B}) = \mathcal{M}_1 \cup \dots \cup \mathcal{M}_n$$  where for $i\in [n]$, $\mathcal{M}_i$ satisfies no $\epsilon(l)$-robust bipartite minor condition involving symbols of arity at most $l$. Then $\PCSP(\mathbb{A}, \mathbb{B})$ is $\mathbf{NP}$-hard.
\end{theorem}

We remark that the $\mathcal{M}_i$ above are not necessarily minions themselves. We now state a generalization of Lemma 5.12 from \cite{AlgApproachToPSCP}.

\begin{theorem}
\cite[cf. Lemma 5.12]{AlgApproachToPSCP}
    Let $\mathcal{M}$ be a set of functions and let $C\colon \mathbb{N} \to \mathbb{N}$. Suppose there exists $I\colon \mathcal{M}\to \mathcal{P}(\mathbb{N})$ such that for all $m,n\in \mathbb{N}\backslash \lbrace 0 \rbrace$
    
    \begin{enumerate}
        \item[(1)] for all $f\in \mathcal{M}^{(n)}$, we have $|I(f)|\leq C(n)$,
        \item[(2)] for all $\pi\colon [n]\to [m]$ and all $g\in \mathcal{M}^{(n)}$ with $g^{(\pi)}\in \mathcal{M}$ we have $$\pi(I(g(x_1, \dots ,x_m)))\cap I(g(x_{\pi(1)}, \dots , x_{\pi(m)}))\neq \emptyset.$$
    \end{enumerate}
 Then for all $l\in \mathbb{N}$, we have that $\mathcal{M}$ satisfies no $1/C(l)^2$-robust bipartite minor condition involving symbols of arity at most $l$.
\end{theorem}

\begin{proof}
    The proof is exactly as in \cite[Lemma 5.12]{AlgApproachToPSCP} in which $\mathcal{M}$ is assumed to be a minion. As mentioned in \cite{AlgApproachToPSCP}, with the above revised hypotheses, this assumption is not necessary.
\end{proof}

\begin{corollary}
\cite[cf. Corollary 5.13]{AlgApproachToPSCP}
\label{1/K^2 Robust}
    Let $\mathcal{M}$ be a set of functions and let $I$ be a selection function for $\mathcal{M}$ with bound $K$. Then $\mathcal{M}$ satisfies no $1/K^2$-robust bipartite minor condition.
\end{corollary}

\begin{theorem}
    \cite[cf. Theorem 5.1]{larrauri2025equationsfinitemonoidsinfinite}
    \label{SelectionFunctionCriteria}
    Let $\mathbb{A}$ and $\mathbb{B}$ be finite relational structures with a homomorphism $\mathbb{A}\to \mathbb{B}$. Suppose that $$\Pol(\mathbb{A}, \mathbb{B}) = \mathcal{M}_1 \cup \dots \cup \mathcal{M}_k$$ and that for each $i\in [k]$, there exists a selection function $I_i$ for $\mathcal{M}_i$. Then $\PCSP(\mathbb{A}, \mathbb{B})$ is $\mathbf{NP}$-hard.
\end{theorem}

\begin{proof}
    By combining Proposition \ref{MinionUnion1} and Corollary \ref{1/K^2 Robust} we have the result.
\end{proof}

\section{Background on semigroups and monoids}\label{Background on semigroups and monoids}

The following background information on semigroups will be used in the proofs of tractability and hardness results in sections \ref{Tractability Results} and \ref{Hardness results}.

\begin{definition}
    \label{Magma}
    A \emph{magma} $\mathbf{M}=(M; \cdot)$ is an algebra with a binary operation $\cdot$. If $\cdot $ is associative, we call $\mathbf{M}$ a \emph{semigroup}. If $\mathbf{M}$ is a semigroup with a two sided identity, we call $\mathbf{M}$ a \emph{monoid}.
\end{definition}

\begin{lemma}
    \label{idempotentElts}
    \cite[cf. Theorem 1.2.3]{Howie}
    Let $\mathbf{S}$ be a finite semigroup. There exists an integer $k>0$ such that $a^k$ is idempotent for all $a\in S$. Moreover, every $a\in S$, has a unique idempotent power denoted $a^{\omega}$.
\end{lemma}

\begin{definition}
    Let $\mathbf{S}$ be a semigroup. An element $a\in S$ is called \emph{group element} if $a$ is a member of a subgroup of $\mathbf{S}$.
\end{definition}

Note that if an element $a$ of a finite semigroup $\mathbf{S}$ is group element, we have $a^k = a^{\omega}$ for some $k>0$, and $a^{\omega}$ is the identity in any subgroup of $\mathbf{S}$ which contains $a$. Because $a^{k-1}a =a^{\omega} = aa^{k-1}$, we see that $a^{k-1}$ is the unique inverse of $a$ in any subgroup containing $a$. Therefore, we may denote the \emph{inverse} of a group element element $a$ of a finite semigroup by $a^{-1}$.

We remark that an element $a$ of a semigroup $\mathbf{S}$ is called \emph{regular} if there exists some $b\in S$ with $aba = a$ (see \cite[Section 2.3]{Howie}). Group elements are always regular. If $\mathbf{S}$ is commutative, it can be shown that regular elements are group elements. Therefore, when working with a commutative semigroup, we will refer to regular elements as group elements. A \emph{regular semigroup} is a semigroup in which every element is regular. Hence if $\mathbf{S}$ is a finite commutative regular semigroup, then every element of $\mathbf{S}$ is a group element. If $\mathbf{S}$ is a finite semigroup and $a\in S$, it is easy to see that $aa^{\omega} = a^\omega a$ is regular. We write $aa^{\omega} = a^{\omega +1}$.

\begin{lemma}
    \label{lemma idempotent product}
    Let $\mathbf{S}$ be a finite commutative semigroup, and let $a,b,c\in S$ such that $c$ is a group element and $c=ab$. Then $c=ca^{\omega}=cb^\omega$.
\end{lemma}

\begin{proof} By commutativity, we have $c^{\omega} = a^{\omega}b^\omega$. Then $c^{\omega}a^{\omega} = a^{\omega}a^{\omega}b^\omega = a^{\omega}b^\omega = c^{\omega}$. Because $c$ is a group element, $c= cc^{\omega} = cc^{\omega}a^{\omega} = ca^{\omega} $. Similarly, $c = cb^\omega$.
\end{proof}

\begin{definition}
    Let $\mathbf{S}$ be a finite commutative semigroup and let $U\subseteq S$. We define $U_\dagger = \lbrace a^{\omega +1} : a\in U \rbrace$.
\end{definition}

\begin{definition}
    Let $\mathbf{S} = (S; \cdot)$ be a semigroup and let $U,V\subseteq S$. We define $U\cdot V = \lbrace u\cdot v : u\in U, v\in V\rbrace .$ Similarly, for $n>0$ we define $U^{(n)} = \lbrace \prod_{i\in [n]}u_i : u_i\in U\rbrace$. When the choice of the operation $\cdot$ is clear, we may write $U\cdot V = UV$.
\end{definition}

\begin{definition}
    Let $\mathbf{S}$ be a semigroup and $U\subseteq S$ be a subset of group elements. We define $U^{-1} = \lbrace a^{-1}\in S : a\in U \rbrace$.
\end{definition}

\begin{definition}
    Let $\mathbf{S}$ be a finite commutative semigroup. A subset $U\subseteq S$ is called a \emph{coset} of $\mathbf{S}$ if $U$ is a subset of group elements and $U\cdot U^{-1} \cdot U = U$.
\end{definition}

\begin{lemma}
\label{lemmaCosets}
    Let $\mathbf{S}$ be a finite commutative semigroup, let $U\subseteq S$ be a subset of group elements. Let $[U]$ be the closure of $U$ under the operation $m(x,y,z) = xy^{-1}z$. Then $[U]$ is the smallest coset containing $U$.
\end{lemma}

\begin{proof}
    We first note that $[U]$ is a set of group elements since products and inverses of group elements of a commutative semigroup are also group elements.

    Next, we claim $[U][U]^{-1}[U] = [U]$. For any $a \in [U]$, note that $a = aa^{-1}a \in [U][U]^{-1}[U]$. Hence $[U]\subseteq [U][U]^{-1}[U]$. On the other hand, we have $[U][U]^{-1}[U] \subseteq [U]$ since $[U]$ is closed under $m$. Therefore $[U][U]^{-1}[U] = [U]$ so that $[U]$ is a coset.

    Now we claim that any coset containing $U$ contains $[U]$. To see this, let $V$ be a coset containing $U$. Then $V=VV^{-1}V$, so that $V$ is closed under $m$. Hence $V= [V]$. Because $U\subseteq V$, we know $[U]\subseteq [V] = V$ as desired.
\end{proof}

We remark that $[U]$ is denoted $\text{Cos}(U)$ in \cite{larrauri2025equationsfinitemonoidsinfinite}.

\begin{lemma}
\label{constant lemma 2}
    \cite[Lemma 5.5]{larrauri2025equationsfinitemonoidsinfinite}
    Let $\mathbf{M}$ be a finite commutative monoid. Then there exists $K(\mathbf{M})\in \mathbb{N}$ such that for all $n\geq K(\mathbf{M})$ and all $R\subseteq M$, we have $[R_\dag]^{(n)}\subseteq R^{(n)}$.
\end{lemma}

\begin{lemma}
\label{lemmaDagger}
    Let $\mathbf{S}$ be a finite commutative semigroup, let $n\in \mathbb{N}\backslash \lbrace 0 \rbrace$, let $A_1, \dots ,A_n \subseteq S$, and let $a\in  \prod_{i\in [n]}A_i$. Then $a^{\omega +1} \subseteq \prod_{i\in [n]}(A_i)_\dagger$.
\end{lemma}

\begin{proof}
    We know $a = \prod_{i\in [n]}a_i$ for some elements $a_i\in A_i$. By commutativity $$a^{\omega +1} = aa^{\omega} = \prod_{i\in [n]}a_i\prod_{i\in [n]}a_i^{\omega} = \prod_{i\in [n]}a_i^{\omega +1} \in \prod_{i\in [n]}(A_i)_\dagger.$$
\end{proof}

\begin{lemma}
\label{lemmaCosetProduct}
    Let $\mathbf{S}$ be a finite commutative semigroup. For all $n\in \mathbb{N}\backslash \lbrace 0 \rbrace$, and for all $A_1, \dots ,A_n \subseteq S$ which are subsets of group elements, we have $[\prod_{i\in [n]}A_i]\subseteq \prod_{i\in [n]}[A_i]$.
\end{lemma}

\begin{proof}
We will use induction. The case $n=1$ is trivial. The case $n=2$ is needed for the inductive argument. Note that $[A_1][A_2]\supseteq A_1A_2$. Also, by commutativity we have $(ab)^{-1} = a^{-1}b^{-1}$ for all $a,b\in S$, so that $$[A_1][A_2]([A_1][A_2])^{-1}[A_1][A_2]=([A_1][A_1]^{-1}[A_1])([A_2][A_2]^{-1}[A_2])=[A_1][A_2].$$ Hence $[A_1][A_2]$ is a coset containing $A_1A_2$, which is to say $[A_1A_2]\subseteq [A_1][A_2]$. 

For the inductive step we assume the hypothesis holds for $n$. Then $$\left[ \prod_{i\in [n+1]}A_i \right] \subseteq \left[ \prod_{i\in [n]}A_i \right][A_{n+1}] \subseteq \prod_{i\in [n+1]}[A_i].$$

\end{proof}

\begin{lemma}
\label{product lemma}
    \cite[cf. Lemma 4.5]{larrauri2025equationsfinitemonoidsinfinite}
    Let $\A = (A; \cdot^{\A})$ be an algebra where $\cdot^{\A}$ is a binary operation, and let $\B = (B; \cdot^{\B})$ be a finite commutative semigroup. Let $F$ be a finite set of homomorphisms $\A\to \B$ and define $g\colon \A\to \B$ by $g(x )= \prod_{f\in F}f(x)$, where the product is $\cdot^{\B}$. Let $R\subseteq A$. Then $$[g(R)_\dagger] \subseteq \prod_{f\in F}[f(R)_\dagger].$$
\end{lemma}

\begin{proof}
    We have $g(R) \subseteq \prod_{f\in F}f(R)$ so that $g(R)_\dagger \subseteq \prod_{f\in F}f(R)_\dagger$ from Lemma \ref{lemmaDagger}. Hence $$[g(R)_\dagger] \subseteq \left[ \prod_{f\in F}f(R)_\dagger \right] \subseteq \prod_{f\in F}[f(R)_\dagger]$$ by Lemma \ref{lemmaCosetProduct}.
\end{proof}

\section{Varieties and commutators}

We introduce some notation and results on general algebras that we need to formulate our results. For more background, see \cite{BurrisSanka} and \cite{CommutatorTheory}.

\begin{definition}
    \label{Variety}
    A class $\mathcal{V}$ of algebras of signature $\sigma$ that is closed under taking homomorphic images, subalgebras, and direct products is called a \emph{variety}.
\end{definition}

For example, the class of all groups, the class of all abelian groups, and the class of all rings are varieties. In fact, each of these varieties has a \emph{Mal'cev term}.

\begin{definition}
    \thlabel{Mal'cev Term}
    Let $\mathcal{V}$ be a variety. We say the ternary term $m$ in the signature of $\mathcal{V}$ is a \emph{Mal'cev term} of $\mathcal{V}$ if for each $\A\in \mathcal{V}$ and $x,y\in A$ we have $$m^{\A}(x,x,y) = y = m^{\A}(y,x,x).$$
\end{definition}

If $\mathcal{V}$ is the variety of all multiplicative groups, then $xy^{-1}z$ is a Mal'cev term for $\mathcal{V}$. We will now define some operations on congruences of algebras which will be used in the proof Theorem \ref{DichotomyCongruenceModular}.

\begin{definition}
\label{congruences}
    Let $\A$ be an algebra. We denote the congruence lattice of $\A$ by $\Con \; \A$ with intersection $\wedge$ and join $\vee$.
\end{definition}

Of particular interest is the situation in which $\A$ is an algebra such that $\Con \; \A$ is modular:

\begin{definition}
    We say a variety $\mathcal{V}$ is \emph{congruence modular} if for all $\A\in \mathcal{V}$ the congruence lattice $\Con \; \A$ satisfies the modular law: $$(\alpha \wedge \beta)\vee (\alpha \wedge \gamma) = \alpha \wedge ((\alpha\wedge \beta) \vee \gamma)$$ for all $\alpha,\beta,\gamma\in \Con \; \A$.
\end{definition}

Every variety with a Mal'cev term is congruence modular (see \cite[Theorem 5.10]{BurrisSanka}). We will prove a dichotomy for $\PEqn(\A, \B)$ in the case of Mal'cev algebras in Theorem \ref{DichotomyMal'cev} and then generalize this result to the case of algebras in a congruence modular variety in Theorem \ref{DichotomyCongruenceModular}. The proof of Theorem \ref{DichotomyCongruenceModular} will require another operation on $\Con \;\A$ known as the \emph{commutator}.

\begin{definition}
\cite[Definition 3.2, Proposition 3.3]{CommutatorTheory}
\label{MatrixSubalgebra}
    Let $\A$ be an algebra, and let $\alpha,\beta \in \Con\; \A$.
    
    \begin{enumerate}
    \item[(1)] We define $M(\alpha, \beta)$ to be the subalgebra of $\A^4$ generated by $$ \begin{pmatrix}
        a & a\\
        a' & a'
    \end{pmatrix}  \quad \text{ and } \quad \begin{pmatrix}
        b & b'\\
        b & b'
    \end{pmatrix}$$
for all $a\;\alpha \; a'$ and $b \; \beta \;b'$.
\item[(2)] We define $[\alpha , \beta]$ to be the smallest $\delta\in \Con \; \A$ such that $a \;\delta \; b \Rightarrow c \;\delta \;d$ for all $\begin{pmatrix}
    a & b \\
    c & d
\end{pmatrix} \in M(\alpha , \beta)$. We say $[\alpha , \beta]$ is the \emph{commutator} of $\alpha$ with $\beta$.

\end{enumerate}
\end{definition}

The common notion of the commutator of two normal subgroups of a group coincides with the above notion if one identifies normal subgroups with their associated congruence. We can use the commutator to define an \emph{abelian} algebra. Solving systems of equations over abelian algebras tends to be easier than solving systems of equations over non-abelian algebras. We will make this statement precise in sections \ref{Tractability Results}, \ref{Hardness results}, and \ref{Congruence Modular Varieties}.

\begin{definition}
    \label{AbelianAlgebra}
    Let $\A$ be an algebra, let $1_\A$ be the total congruence on $\A$ and let $0_\A$ be the equality relation on $\A$. We say $\A$ is \emph{abelian} if $[1_\A, 1_\A] = 0_\A$.
\end{definition}

It follows from Definition \ref{MatrixSubalgebra} and Definition $\ref{AbelianAlgebra}$ that $\A/[1_\A, 1_\A]$ is abelian. Also, if $\A/\delta$ is abelian, then $\delta \geq [1_\A, 1_\A]$. The converse is not always true, but we will later see that it is true in congruence modular varieties. The notion of an abelian group coincides with our above definition as well. An abelian ring is actually just a ring with 0-multiplication:

\begin{example}
\label{AbelianRing}
    A ring $\A$ is abelian if and only if it has 0-multiplication.
\end{example}

\begin{proof}
    First, suppose $\A$ is abelian. For all $a,b\in A$ we have $$\begin{pmatrix}
        0 & 0\\
        a & a
    \end{pmatrix} , \begin{pmatrix}
        0 & b\\
        0 & b
    \end{pmatrix}\in M(1_\A, 1_\A).$$
Hence $$\begin{pmatrix}
    0 & 0 \\
    0 & ab
\end{pmatrix}\in M(1_\A, 1_\A).$$
Since $[1_\A , 1_\A] = 0_\A$, we have by definition of the commutator that $0=0 \Rightarrow 0 = ab$. That is, $\A$ has 0-multiplication.

For the converse, note that a ring with 0-multiplication is term equivalent to an abelian group.\end{proof}

We will also use the notion of a \emph{solvable} algebra. This coincides with the common notion of a solvable group if our algebra is a group.

\begin{definition}
    \label{Solvable}
    Let $\A$ be an algebra and $\alpha \in \Con \;\A$. Define $[\alpha]^0 = \alpha$ and for $k\geq 0$ define $[\alpha]^{k+1}=[\alpha^k , \alpha^k]$. If $[1_\A]^d = 0_\A$ for some $d\geq 0$, we say, we say $\A$ is $d$-\emph{step solvable} or that $\A$ is \emph{solvable}.
\end{definition}

\begin{theorem}\cite[Theorem 6.2]{CommutatorTheory}
    \label{SolvableInConMod}
    Solvable algebras in a congruence modular variety have a Mal'cev term.
\end{theorem}

We will now list some common properties of the commutator which will be used throughout Section \ref{Congruence Modular Varieties}.

\begin{theorem}\cite[Proposition 3.4, Proposition 4.4]{CommutatorTheory}
    \label{CommMonotone}
    Let $\A$ be an algebra with subalgebra $\C$ and let $\alpha,\beta \in \Con \;\A$. Then 
    \begin{enumerate}
    \item[(1)] $[\alpha , \beta]$ is monotone in both $\alpha$ and $\beta$,
    
    \item[(2)] $[\alpha , \beta] \leq \alpha \wedge \beta$,

    \item[(3)] $[\alpha|_\C , \beta|_\C] \leq [\alpha , \beta]|_\C$.
    \end{enumerate}
\end{theorem}

We remark that Proposition \ref{CommMonotone} (3) implies subalgebras of abelian algebras are abelian.

\begin{theorem}\cite[Proposition 4.3, Proposition 4.4]{CommutatorTheory}
    \label{CommMeetJoin}
    Let $\A , \B$ be algebras in a congruence modular variety, let $\alpha,\beta , \gamma\in \Con \;\A$, and let $f\colon \A\to \B$ be a surjective homomorphism. Then
    
    \begin{enumerate}
        \item[(1)] $[\alpha , \beta] = [\beta, \alpha]$,
    
        \item[(2)] $[\alpha , \beta \vee \gamma] = [\alpha , \beta]\vee [\alpha , \gamma]$,

        \item[(3)] $f([\alpha, \beta] \vee \ker f) = [f(\alpha), f(\beta)]$,
        
    \end{enumerate}
\end{theorem}

For algebras in congruence modular varieties, we remark that Proposition \ref{CommMeetJoin} (3) implies that quotients of abelian algebras are abelian and that $\A / \alpha$ is abelian if and only if $\alpha \geq [1_\A, 1_\A]$. Similarly, quotients of solvable algebras are solvable, and $\A / \alpha$ is solvable if and only if there exists some $d\in \mathbb{N}$ such that $\alpha \geq [1]^d$.

\begin{theorem}
    \cite[cf. Proposition 4.5]{CommutatorTheory}
    \label{CommDirectProd}
    Let $n\in \mathbb{N}\backslash \lbrace 0 \rbrace$, let $\A = \prod_{i\in [n]}\A_i$, let $\A_i$ be in a congruence modular variety for each $i\in [n]$, and let $\alpha_i, \beta_i \in \Con \; \A_i$ for each $i\in [n]$. Then $$\left[\prod_{i\in [n]}\alpha_i , \prod_{i\in [n]}\beta_i \right] = \prod_{i\in [n]}[\alpha_i , \beta_i].$$
\end{theorem}

Here the product is the direct product and the congruence $\prod_{i\in [n]}[\alpha_i, \beta_i]$ is a congruence on the algebra $\prod_{i\in [n]} \A_i$.

\begin{definition}
\label{WeakDifference}
\cite{ALVIII}
    Let $\mathcal{V}$ be a variety. We say the ternary term $d$ in the signature of $\mathcal{V}$ is a \emph{weak difference term} of $\mathcal{V}$ if $d$ is idempotent and for each $\A\in \mathcal{V}$ and $x,y \in A$ we have $$d^{\A}(x,x,y)\; \alpha \; y \;  \alpha \; d^{\A}(y,x,x)$$ where $\alpha\in \Con \;\A$ is given by the commutator $\alpha = [\Theta(x,y), \Theta(x,y)]$ and $\Theta(x,y)$ is the congruence generated by the pair $(x,y)$.
\end{definition}

It is known that congruence modular varieties have weak difference terms. In fact if $\mathcal{V}$ is a congruence modular variety, then $\mathcal{V}$ has a \emph{difference term} $d$. For all $\A \in \mathcal{V}$ and for all $x,y\in A$, the term $d$ satisfies $$d^\A(x,x,y) = y \; \alpha \; d^\A(y,x,x)$$ where $\alpha = [\Theta(x,y), \Theta(x,y)]$ (see \cite[Theorem 5.5]{CommutatorTheory}).

Note that in an abelian algebra $\A$, every commutator is the equality relation $0_\A$. In particular, for all $x,y\in A$ we know that $\alpha = [\Theta(x,y), \Theta(x,y)]$ is the equality relation. That is, a weak difference term of an abelian algebra is just a Mal'cev term. Furthermore, in the case where $\A$ is abelian, it is well known that the Mal'cev term operation $d^{\A}$ can be given by $d^{\A}(x,y,z) = x-y+z$ where $+$ is addition of an abelian group on the set $A$ and that $d^{\A}$ commutes with the basic operations of the algebra $\A$ (see \cite[Lemma 5.6, Proposition 5.7]{CommutatorTheory}).

\section{Tractability results}\label{Tractability Results}

We will first extend the tractability part of \cite[Theorem 3]{larrauri2024solving} by mimicking the proof technique in a different setting, in particular for algebras in varieties with weak difference terms.

\begin{theorem}
    \label{TractabilityTaylorAlg}
    Let $\A$ and $\B$ be finite algebras in a variety with a weak difference term $d$ such that there exists a homomorphism $\A\to \B$.
    
    If there is a homomorphism $\psi\colon\A \to \B$ such that $\psi(\A) $ is an abelian algebra, then $\PEqn(\A, \B)$ is solvable by AIP, hence in polynomial time.
\end{theorem}

\begin{proof}
    We will apply Proposition \ref{AIP characterization} to get the result, so we need to define for each odd arity a polymorphism from $\A$ to $\B$ which is alternating. Define $\mathbf{C} = \psi(\A)$. Because $\mathbf{C}$ is abelian, we have that $d^{\mathbf{C}}$ is a Mal'cev term operation with $d^{\mathbf{C}}(x,y,z) = x-y+z$ where $+$ is the addition in an abelian group from \cite[Lemma 5.6]{CommutatorTheory}. Furthermore, $d^{\mathbf{C}}$ commutes with all the basic operations of $\mathbf{C}$ by \cite[Proposition 5.7]{CommutatorTheory}.

    For $n\geq 0$, we define $$p_{2n+1}\colon A^{2n+1} \to B, \quad (a_1, \dots , a_{2n+1}) \mapsto  \psi(a_1)-\psi(a_2)+\dots -\psi(a_{2n})+\psi(a_{2n+1}).$$

    From this it is clear we can permute inputs with even indices or with odd indices and not change the value of $p_{2n+1}$. Furthermore, for any $a_1, \dots , a_{2n-1}, x, y \in A$, it is easy to see that $$p_{2n+1}(a_1, \dots , a_{2n-1}, x, x) = p_{2n-1}(a_1, \dots , a_{2n-1}) = p_{2n+1}(a_1, \dots , a_{2n-1}, y, y).$$
    
    We now show that $p_{2n+1}$ is a polymorphism from $\A$ to $\B$. Let $f$ be a basic operation symbol in $\sigma$ of arity $k$. Let $\overline{a}_i \in A^k$ for each $i\in [2n+1]$ with $\overline{a}_i = (a_1^i, \dots a_k^i)$. Then
    \begin{align*}
        &p_{2n+1}(f^\textbf{A}(\overline{a}_1) \dots , f^\textbf{A}(\overline{a}_{2n+1}))\\  & =
        \psi(f^\textbf{A}(\overline{a}_1))-\psi(f^\textbf{A}(\overline{a}_2))+\dots -\psi(f^\textbf{A}(\overline{a}_{2n}))+\psi(f^\textbf{A}(\overline{a}_{2n+1}))\\
        & =f^\textbf{C}(\psi(a_1^1), \dots , \psi (a_k^1)) - \dots + f^\textbf{C}(\psi(a_1^{2n+1}), \dots , \psi(a_k^{2n+1})) \quad \text{ as }\psi \text{ is a homomorphism}\\
        & =f^\textbf{C}(\psi(a_1^1)-\dots + \psi(a_{1}^{2n+1}), \dots  ,\psi(a_k^1)-\dots + \psi(a_{k}^{2n+1}) ) \quad \text{ as } d^{\mathbf{C}} \text{ commutes with }f^\textbf{C}\\
        & = f^\textbf{C}(p_{2n+1}(a_1^1, \dots , a_{2n+1}^1), \dots , p_{2n+1}(a_k^1, \dots , a_{k}^{2n+1})) \quad \text{ as desired.}
    \end{align*}

\end{proof}

The following lemma will be useful for reducing excess notation.

\begin{lemma}
\label{lemma2}
     Let $\A$ and $\B$ be algebras of the same signature with a polymorphism $p\colon \,\A^n\to \B$, let $f^{\A}$ an $k+2$-ary term operation on $\A$, let $c_i \in A$ for each $i\in[k]$, and let the binary polynomial operation $\cdot$ on $\A$ be given by $$x\cdot y= f^{\A}(x,y,c_1, \dots , c_k).$$
Define the binary polynomial operation $\cdot_p$ on $p(\A^n)$ by $$x\cdot_p y= f^{\B}(x,y,\dots p(c_1, \dots ,c_1), \dots , p(c_k, \dots c_k)).$$
Then $p\colon (A; \cdot)^n \to (p(A^n); \cdot_p)$ is a homomorphism.
\end{lemma}

\begin{proof}
    By definition of a polymorphism we have 
        \begin{align*}
    p(x_1\cdot y_1, \dots ,x_k\cdot y_k) &= p(f^{\A}(x_1, y_1, c_1, \dots c_{k}), \dots , f^{\A}(x_n, y_n, c_1, \dots c_{k}) )\\ &= f^{\B}(p(x_1, \dots , x_n), p(y_1, \dots , y_n), p (c_1, \dots , c_1), \dots ,p(c_{k}, \dots c_{k}))\\
    &=p(x_1, \dots , x_k)\cdot_p p(y_1, \dots ,y_k).
     \end{align*}\end{proof}

The next result involves algebras expanded with an additional relation. In this setting, it is useful to consider the relational structures corresponding to our algebras as in Proposition \ref{PEqnasPCSP}. If $\mathbb{A} = (A; R_1^\mathbb{A}, \dots , R_n^{\mathbb{A}})$ and $\mathbb{B}=(B; R_1^\mathbb{B}, \dots , R_n^{\mathbb{B}})$ are relational structures with a homomorphism $\mathbb{A}\to \mathbb{B}$, it can be shown that $\PCSP(\mathbb{A}, \mathbb{B})$ is polynomially equivalent to $$\PCSP((A; R_1^\mathbb{A}\times \dots \times R_n^{\mathbb{A}}), (B; R_1^\mathbb{B}\times \dots \times R_n^{\mathbb{B}})).$$ That is, every PCSP can be viewed as a PCSP involving relational structures with just one relation.

\begin{theorem}
    \label{TractabilityBinaryPolOp}
    Let $\A=(A; f_1^{\A}, \dots , f_r^{\A})$ and $\B= (B; f_1^{\B}, \dots , f_r^{\B})$ be finite algebras, let $\cdot$ be a binary polynomial operation on $A$. Let $R^{\A}\subseteq A^r$ and $R^\B\subseteq B^r$ be relations, and define $S^{\A} = (f_1^{\A})^\circ\times \dots \times(f_r^{\A})^\circ\times R^{\A}$, $S^{\B} = (f_1^{\B})^\circ\times \dots \times(f_r^{\B})^\circ\times R^{\B}$, $\overline{\A} = (A; S^{\A}), \overline{\B} = (B; S^{\B})$.
    
    If there is a homomorphism $\psi\colon \overline{\A} \to \overline{\B}$ such that
    \begin{enumerate}

    \item[(1)] $( \psi(A);\; \cdot_\psi )$ is a commutative regular semigroup and

    \item[(2)] $[\psi(S^{\A})]\subseteq S^{\B}$,

    \end{enumerate}

    then $\PCSP(\overline{\A}, \overline{\B} )$ is solvable by BLP+AIP, hence is in $\mathbf{P}$.
    
\end{theorem}

\begin{proof}
    Let $\psi \colon \overline{\A}\to \overline{\B}$ be a homomorphism that satisfies conditions 1. and 2.

Now we can closely follow an argument from \cite[Proposition 6.2]{larrauri2025equationsfinitemonoidsinfinite}. We will show there exist 2-block symmetric polymorphisms of every odd arity from $\overline{\A}$ to $\overline{\B}$ and then apply Proposition \ref{BLP+AIP characterization}.
    
Because $(\psi(A); \cdot_\psi ) $ is a union of (finite) abelian groups by assumption, we know that for all $x\in \psi(A)$, there exists a unique $x^{-1}\in \psi(A)$ which is an inverse of $x$ with respect to the multiplication $\cdot_\psi$. We define for each $n\geq 0$ $$p_{2n+1}\colon A^{2n+1} \to B, \quad (a_1, \dots a_{2n+1})\mapsto \prod_{\substack{i\in [2n+1]\\i \text{ odd}}}\psi(a_i)\prod_{\substack{j\in [2n+1]\\j \text{ even}}}\psi(a_j)^{-1},$$
where the multiplication above is $\cdot_\psi$.

Since our multiplication is commutative, it is immediate that $p_{2n+1}$ is 2-block symmetric with blocks corresponding to the even and odd indices. We remark that $p_{2n+1}$ is not necessarily alternating (although $\psi(x)\psi(x)^{-1}$ is an identity of a subgroup of $\psi(A)$, it need not be an identity of $\psi(A)$). Next, we need to show that $p_{2n+1}$ is a polymorphism from $\A$ to $\B$. That is, we must show that $p_{2n+1}((S^{\A})^{2n+1})\subseteq S^\B$. We have
\begin{align*}
p_{2n+1}((S^{\A})^{2n+1}) &= \psi(S^{\A})^{ (n+1)}\cdot_\psi (\psi(S^{\A})^{-1})^{( n)}\\
&=\psi(S^{\A})\cdot_\psi \left( (\psi(S^{\A}))^{-1} \cdot_\psi\psi(S^{\A}) \right)^{(n)}\\
&\subseteq [\psi(S^{\A})]\\
&\subseteq S^{\B}.
\end{align*}

So $p_{2n+1}$ is a 2-block symmetric polymorphism from $\overline{\A}$ to $\overline{\B}$ for each $n\geq 0$ and therefore $\PCSP(\overline{\A}, \overline{\B})$ is in $\mathbf{P}$ by Proposition \ref{BLP+AIP characterization}. \end{proof}

\section{Hardness results}\label{Hardness results}

In \cite{larrauri2025equationsfinitemonoidsinfinite}, Larrauri, Mottet, and Živný showed that there is a dichotomy for $\PEqn(\A, \B)$ where $\A$ and $\B$ are finite monoids expanded with an additional relation. We can use the proof techniques from \cite[Theorem 3]{larrauri2024solving} and \cite[Proposition 5.9]{larrauri2025equationsfinitemonoidsinfinite} to generalize this dichotomy to the case in which $\A$ and $\B$ have a binary polynomial operation and this operation has a left and right identity in $\A$. We remark that this binary polynomial operation need not be associative.

The following is contained in the proof of \cite[Theorem 5]{larrauri2024solving}, and will be applied later in the construction of a selection function in the proof of our main hardness result.

\begin{theorem}
\cite{larrauri2024solving}
\label{SelectionCorollary}
    Let $\mathbf{M}$ be a finite monoid with $a\in M$ such that $a$ is not a group element of $\mathbf{M}$. Then there exists $K\in \mathbb{N}$ and a selection function $I\colon \mathcal{M}_{\mathbf{M}, a} \to \mathcal{P}(\mathbb{N})$ with bound $K$.
\end{theorem}

\begin{theorem}
    \label{HomomorphismToUnionOfMonoidalMinions}
    Let $\mathcal{M} , \mathcal{N}$ be minions. If $\xi\colon \mathcal{M}\to\mathcal{N}$ is a minion homomorphism and $\mathcal{N}$ has a selection function $I$ with bound $K$, then $I\circ\xi$ is a selection function for $\mathcal{M}$ with bound $K$.
\end{theorem}

\begin{proof}
    Since $I$ has bound $K$, it is immediate that $I\circ \xi$ has bound $K$. Let $p\in \mathcal{M}^{(n)}$ and let $\pi\colon[n]\to [m]$ be a function. We have 
    \begin{align*}
        \pi(I(\xi(p)) )\cap I((\xi(p^{(\pi)}))&=\pi(I(\xi(p)) )\cap I((\xi(p))^{(\pi)})\neq \emptyset .
    \end{align*}\end{proof}

We will use the $p_i$ functions in the next lemma in the proof of Theorem \ref{MainDichotomy}.

\begin{lemma}
\label{p_i basic properties}
    Let $\A$ and $\B$ be finite algebras, let $e\in A$, let $\cdot$ be a binary polynomial operation on $\A$ such that $x\cdot e = x = e\cdot x$ for all $x\in A$, let $p\in \Pol(\A, \B)^{(n)}$ where $n\in \mathbb{N}\backslash \lbrace 0 \rbrace$, and for each $i\in [n]$, let $a_i\in A$. Define $$p_{i}\colon A\to B, \quad x \mapsto p(e, \dots ,e,\underset{i}{x},e,\dots e).$$ Then
\begin{enumerate}
    \item[(1)] $p_i(a_i)p_j(a_j) = p_j(a_j)p_i(a_i)$ for $1\leq i < j \leq n$,
    \item[(2)] $p_i(a_i) \big( p_j(a_j)p_k(a_k) \big) = \big(p_i(a_i)p_j(a_j) \big)p_k(a_k)$ for $1\leq i < j < k \leq n$, 
    \item[(3)] $p(a_1, \dots ,a_n) = p_1(a_1)\dots p_n(a_n)$
\end{enumerate}

where the multiplication is $\cdot_p$.
\end{lemma}

\begin{proof}
    Throughout, the multiplication occurring in $p(A^n)$ is $\cdot_p$, so we will omit the multiplication symbol for readability. To prove (1), observe that 
\begin{align*}
    p_i(a_i)p_j(a_j) &= p(e, \dots ,e , \underset{i}{a_i}, e,\dots e)p(e, \dots ,e , \underset{j}{a_j}, e,\dots e) \\
    & = p(e, \dots ,e , \underset{i}{a_i}, e, \dots ,e , \underset{j}{a_j},e \dots ,e)\\
    &=p(e, \dots ,e , \underset{j}{a_j}, e,\dots e)p(e, \dots ,e , \underset{i}{a_i}, e,\dots e)\\
    &=p_j(a_j)p_i(a_i)
\end{align*}    
    by applying Lemma \ref{lemma2}.

To prove (2), we similarly observe that 
\begin{equation*}
    p_i(a_i)\big(p_j(a_j)p_k(a_k)\big) = p(e, \dots , e , \underset{i}{a_i}, e , \dots , e , \underset{j}{a_j}, e , \dots , e , \underset{k}{a_k}, e ,\dots ,e) = \big(p_i(a_i)p_j(a_j)\big)p_k(a_k).
\end{equation*}

The proof of (3) is similar.\end{proof}

\begin{lemma}
    \label{claim 2}
     Let $\A$ and $\B$ be finite algebras, let $e\in A$, let $\cdot$ be a binary polynomial operation on $\A$ such that $x\cdot e = x = e\cdot x$ for all $x\in A$, let $p\in \Pol(\A, \B)^{(n)}$ where $n\in \mathbb{N}\backslash \lbrace 0 \rbrace$. Let $1\leq i<j<k \leq n$. We have:
     
     \begin{enumerate}
         \item[(1)] if $p_i(A) = p_j(A) = p_k(A)$, then $(p_i(A); \cdot_p )$ is a commutative monoid.
         \item[(2)] $(p(A^n) ; \cdot_p)$ is a commutative monoid if and only if $(p_i(A); \cdot_p)$ is a commutative monoid for each $i\in [n]$.
     \end{enumerate}
\end{lemma}

\begin{proof}
    (1) is immediate from Lemma \ref{p_i basic properties}.

    To prove (2), first note that if $(p(A^n) ; \cdot_p)$ is a commutative monoid, then each subalgebra $(p_i(A); \cdot_p)$ is a commutative monoid.

    For the converse, we first show commutativity of $(p(A^n) ; \cdot_p)$. Let $a_1, \dots , a_n ,b_1,\dots ,b_n\in A$ and consider $p(a_1, \dots , a_n), p(b_1, \dots , b_n)\in p(A^n)$. We get 
\begin{align*}
    &p(a_1, \dots , a_n)\cdot_p p(b_1, \dots , b_n)\\& = p(a_1\cdot b_1, \dots , a_n\cdot b_n) \text{ by Lemma \ref{lemma2}}\\
    &=p_1(a_1\cdot b_1)\cdot_p \dots \cdot_p  p_n(a_n\cdot b_n) \text{ by Lemma \ref{p_i basic properties} (3)}\\
    &=\left(p_1(a_1)\cdot_p p_1(b_1) \right)\cdot_p \dots \cdot_p  (p_n(a_n)\cdot_p p_n(b_n)) \text{ by Lemma \ref{lemma2}}\\
    &= \left(p_1(b_1)\cdot_p p_1(a_1) \right)\cdot_p \dots \cdot_p  (p_n(b_n)\cdot_p p_n(a_n)) \text{ since }(p_i(A); \cdot_p) \text{ is commutative for each }i\in [n]\\
    &=p_1(b_1\cdot a_1)\cdot_p \dots \cdot_p  p_n(b_n\cdot a_n) \text{ by Lemma \ref{lemma2}}\\
    &= p(b_1\cdot a_1, \dots , b_n\cdot a_n)\text{ by Lemma \ref{p_i basic properties}} (3)\\
    &=p(b_1, \dots , b_n)\cdot_p p(a_1, \dots , a_n).
\end{align*}

For associativity, the proof is similar.
\end{proof}

In the proof of Theorem \ref{MainDichotomy}, we will construct a selection function which outputs unions of equivalence classes given in the following definition.

\begin{definition}
\cite{larrauri2025equationsfinitemonoidsinfinite}
    \label{ConstantSet}
    Let $\A$ and $\B$ be algebras, let $e\in A$, let $\cdot $ be a binary polynomial operation on $A$ such that for all $x\in A$, we have $x\cdot e = x = e\cdot x$, and let $p\in \Pol(\A, \B)^{(n)}$. We define a binary relation $\sim_p$ on $[n]$ by letting $i\sim_p j$ if $p_i = p_j$.
    
    Note that $\sim_p$ is an equivalence relation. We will refer to the equivalence classes of $\sim_p$  as $p$-\emph{equivalence classes}.
\end{definition}

The next important lemma is a generalization of \cite[Proposition 5.8]{larrauri2025equationsfinitemonoidsinfinite} and will help us construct a selection function for $\Pol(\A, \B)$ in our proof of Theorem \ref{DichotomyBinaryTermOp}

\begin{lemma}
\thlabel{biglemma}
    \cite[cf. Proposition 5.8]{larrauri2025equationsfinitemonoidsinfinite}
    Let $\A$ and $\B$ be finite algebras of the same signature, let $\cdot $ be a binary polynomial operation on $\A$, let $e\in A$ be such that $e\cdot x = x = x \cdot e$ for all $x\in A$. Let $\mathcal{M} \subseteq \Pol(\A, \B)$ be the subset of polymorphisms $p$ such that:

    \begin{enumerate}
        \item[(1)] $\mathbf{M}=(p(A^{\ar(p)}); \cdot_p )$ is a commutative monoid,
        \item[(2)] $\mathbf{N}=(p_{\Delta}(A); \cdot_{p_{\Delta}})$ is a regular submonoid of $\mathbf{M}$.
    \end{enumerate}

Let $r\in \mathbb{N}\backslash \lbrace 0 \rbrace$, and let $R\subseteq A^r$. For each $p\in \mathcal{M}$, define $I(p)$ to be the union of $p$-equivalence classes of size smaller than $$L= \max\lbrace K( \mathbf{C}^r) : \mathbf{C}\text{ is a commutative monoid with } |C|\leq |B| \rbrace$$ where $K(\mathbf{C}^r)$ is obtained from Lemma \ref{constant lemma 2}.

If $I(p^{(\pi)} ) \cap \pi(I(p)) = \emptyset$ for some $p, p^{\pi} \in \mathcal{M}$, then $[p_\Delta(R)]\subseteq  p(R^{\ar (p)})$.
\end{lemma}

\begin{proof}
    We follow the proof techniques from \cite{larrauri2025equationsfinitemonoidsinfinite}. Note that our multiplications $\cdot$, $\cdot_p$ are polynomial operations on our algebras rather than basic operations.
    
    Let $p,q\in \mathcal{M}$ with $\ar(p) = n$, $\ar(q) = m$, and $\pi\colon [n]\to [m]$ be such that $q=p^{(\pi)}$ and $I(p^{(\pi)} ) \cap \pi(I(p)) = \emptyset$.

    Note that for each $i\in [n]$, the map $p_i\colon (A; \cdot) \to (p(A^n); \cdot_p)$ given by $p_i(x) = p(e, \dots , \underset{i}{x}, e , \dots , e)$ is a homomorphism and that the map $p_\Delta\colon \A\to \B$ satisfies $p_{\Delta}(x)= p(x, \dots , x) = p_1(x)\cdot_p \dots \cdot_p p_n(x)$ by Lemma \ref{lemma2}. Corresponding statements hold for $q_j$ for each $j\in[m]$ and $q_\Delta$. Since $q$ is minor of $p$, we have $p_\Delta= q_\Delta$, and $\cdot_q$ is a restriction of $\cdot_p$. By the regularity of $\mathbf{N}$, we know that $q_\Delta(R) = q_\Delta(R)_\dagger$. Therefore
\begin{equation}
        \label{eq1}
        [p_\Delta(R)] = [q_\Delta(R)] = [q_\Delta(R)_\dagger] \subseteq \prod_{i\in [m]}[q_i(R)_\dagger] 
    \end{equation}
    by Lemma \ref{product lemma}, where the product is $\cdot_q$. Moreover
\begin{equation}
        \label{eq2}
        \prod_{i\in [m]}[q_i(R)_\dagger] = \prod_{i\in I(q)}[q_i(R)_\dagger] \cdot_q \prod_{i\in [m]\backslash I(q)}[q_i(R)_\dagger].
    \end{equation}
Note by definition of $I$ that $[m]\backslash I(q)$ is a union of $q$-equivalence classes of size greater than or equal to $L$. Let $J \subseteq [m]\backslash I(q)$ be one of these equivalence classes. Then $q_i=q_j$ for all $i,j\in J$ and we may define $q_J = q_i$ where $i\in J$. Hence 
    \begin{equation}
    \label{eq3}
        \prod_{i\in J}[q_i(R)_\dagger] = [q_J(R)_\dagger]^{|J|} \subseteq q_J(R)^{|J|} = \prod_{i\in J}q_i(R)
    \end{equation}
    by Lemma \ref{constant lemma 2}. 
Combining \ref{eq2} and \ref{eq3}, we have
\begin{equation}
        \label{eq4}
        \prod_{i\in [m]}[q_i(R)_\dagger] = \prod_{i\in I(q)}[q_i(R)_\dagger] \cdot_q \prod_{i\in [m]\backslash I(q)}q_i(R).
    \end{equation}
Now recall that $q = p^{(\pi)}$ so that for $i\in [m]$, we have $q_i(x) = \prod_{j\in \pi^{-1}(i)}p_j(x)$ for all $x\in R$. By Lemma \ref{product lemma} we get $[q_i(R)_\dagger]\subseteq \prod_{j\in \pi^{-1}(i)}[p_j(R)_\dagger]$. Now combining this with \ref{eq1} and \ref{eq4}, we have \begin{equation}
        \label{eq5}
        [p_\Delta(R)]\subseteq \prod_{j\in \pi^{-1}(I(q))}[p_j(R)_\dagger]\cdot_p\prod_{j\in [n]\backslash \pi^{-1}(I(q))}p_j(R).
    \end{equation}

    We want to replace the right hand side of \ref{eq5} with $\prod_{j\in [n]}p_j(R)$. To do this, let $c\in [p_\Delta(R)]$ so that $c= \prod_{j\in [n]}a_j$ where
\[a_j\in
    \begin{cases}
       [p_j(R)_\dagger] &\text{ if }j\in \pi^{-1}(I(q)),\\
       p_j(R) &\text{ otherwise.}
    \end{cases}\]
Because $c$ is a group element, we can apply Lemma \ref{lemma idempotent product} to get $ca_j^{\omega} = c$ for all $j\in [n]$. In particular, this implies that
\begin{equation*}
        \label{eq7}
        c = \prod_{j\in I(p)}a_j \prod_{j\in [n]\backslash I(p)}a_ja_j^{\omega}.
    \end{equation*}
For $j\in [n]\backslash I(p)$, we have $a_ja_j^{\omega}\in [p_j(R)_\dagger]$. Also, note that $[n]\backslash I(p)$ is a union of $p$-equivalence classes of size greater than or equal to $L$ by definition of $I$. Let $J\subseteq [n]\backslash I(p)$ be such an equivalence class. Then by Lemma \ref{constant lemma 2}, we have \begin{equation*}
        \label{eq8}
        \prod_{i\in J}[p_i(R)_\dagger] = [p_J(R)_\dagger]^{|J|}\subseteq p_J(R)^{|J|} = \prod_{i\in J}p_i(R).
    \end{equation*}

    By assumption, we have $\pi(I(p))\cap I(q) = \emptyset$ so that $I(p)\cap \pi^{-1}(I(q)) = \emptyset$. Then for $j\in I(p)$, we know that $a_j \in p_j(R)$. We have shown that $$c \in \prod_{j\in I(p)}p_j(R) \prod_{j\in [n]\backslash I(p)}p_j(R) = \prod_{j\in [n]}p_j(R) = p(R^n)  .$$ Hence $[p_\Delta(R)]\subseteq p(R^n)$ as desired.
\end{proof}

In order to state our main results in more familiar language, we will prove the following lemma.

\begin{lemma}
    \label{lemma cosets and polymorphisms}
    Let $\A=(A; f_1^{\A}, \dots , f_r^{\A})$ and $\B= (B; f_1^{\B}, \dots , f_r^{\B})$ be finite algebras, let $\cdot$ be a binary polynomial operation on $\A$. Define $S^{\A} = (f_1^{\A})^\circ\times \dots \times(f_r^{\A})^\circ $, $S^{\B} = (f_1^{\B})^\circ\times \dots \times(f_r^{\B})^\circ $. Suppose there is a homomorphism $\psi\colon {\A} \to {\B}$ such that $(\psi(A);\cdot_\psi)$ is a commutative regular monoid.

    Then $[\psi(S^{\A})]\subseteq S^{\B}$ if and only if $m(x,y,z)=x\cdot_\psi y^{-1}\cdot_\psi z$ is a polymorphism of the algebra $\psi(\A)$.
\end{lemma}

\begin{proof}
    We first assume $[\psi(S^{\A})]\subseteq S^{\B}$ and show that $m$ preserves $f_i^{\psi(\A)}$ for each $i\in[n]$. So pick $i\in[n]$ and note that our assumption implies that $[(f_i^{\mathbf{\psi(A)}})^\circ]\subseteq (f_i^{\B})^\circ$. Recall by Lemma \ref{lemmaCosets} that $[\psi(S^{\A})]$ is the closure of $\psi(S^{\A})$ under the operation $m$. In particular, we know $[(f_i^{\mathbf{\psi(A)}})^\circ]$ is the closure of $(f_i^{\mathbf{\psi(A)}})^\circ$ under the operation $m$. Therefore $$m\big((f_i^{\mathbf{\psi(A)}})^\circ, (f_i^{\mathbf{\psi(A)}})^\circ , (f_i^{\mathbf{\psi(A)}})^\circ  \big) \subseteq (f_i^{\B})^\circ .$$ Let $k=\ar(f_i)$, and let 
    \begin{equation}
    \label{eq9}
        (x_1, \dots ,x_k,x_{k+1}), (y_1, \dots ,y_k,y_{k+1}), (z_1, \dots ,z_k,z_{k+1})\in (f_i^{\psi(\A)})^\circ .
    \end{equation}
    Then we have 
\begin{equation}
    \label{eq10}
        \left( m(x_1,y_1,z_1 ), \dots , m(x_k,y_k,z_k), m(x_{k+1}, y_{k+1}, z_{k+1}) \right) \in (f_i^{\B})^\circ .
    \end{equation}
But each of the $m(x_j,y_j,z_j)$ for $j\in [k+1]$ are members of $\psi(A)$ as $m$ is an operation on $\psi(A)$. So \ref{eq10} says that
\begin{equation}
        \label{eq11}
        f_i^{\psi(\A)}(  m(x_1,y_1,z_1 ), \dots , m(x_k,y_k,z_k) ) = m(x_{k+1}, y_{k+1}, z_{k+1}).
    \end{equation}
But by \ref{eq9} $f_i^{\psi(\A)}(x_1, \dots , x_k) = x_{k+1}$, $f_i^{\psi(\A)}(y_1, \dots , y_k) = y_{k+1}$, and $f_i^{\psi(\A)}(z_1, \dots , z_k) = z_{k+1}$, so we have shown that $m$ commutes with $f_i^{\psi({\A)}}$ as desired.
    
    Conversely, if $m$ is a polymorphism of the algebra $\psi(\A)$, then we have $$m\big((f_i^{\mathbf{\psi(A)}})^\circ, (f_i^{\mathbf{\psi(A)}})^\circ , (f_i^{\mathbf{\psi(A)}})^\circ  \big) \subseteq (f_i^{\mathbf{\psi(\A)}})^\circ $$ for each $i\in [n]$. Therefore $(f_i^{\psi(\A})^\circ$ is closed under $m$ for each $i\in [n]$. By Lemma \ref{lemmaCosets}, $$[(f_i^{\psi(\A})^\circ] = (f_i^{\psi(\A})^\circ \subseteq (f_i^{\B})^\circ$$ for each $i\in [n]$, which implies $[\psi(S^{\A})]\subseteq S^{\B}$.
\end{proof}

We are now ready to state our first main result. This is a dichotomy theorem for a class of structures which contains expansions of monoids and is a generalization of \cite[Proposition 5.9]{larrauri2025equationsfinitemonoidsinfinite}. 

\begin{main}
    \label{DichotomyBinaryTermOp}

    Let $\A=(A; f_1^{\A}, \dots , f_r^{\A})$ and $\B= (B; f_1^{\B}, \dots , f_r^{\B})$ be finite algebras, let $\cdot$ be a binary polynomial operation on $\A$ such that $x \cdot e = x= e \cdot x$ for all $x\in A$. Let $R^{\A}\subseteq A^r$ and $R^\B\subseteq B^r$ be relations, and define $S^{\A} = (f_1^{\A})^\circ\times \dots \times(f_r^{\A})^\circ\times R^{\A}$, $S^{\B} = (f_1^{\B})^\circ\times \dots \times(f_r^{\B})^\circ\times R^{\B}$, $\overline{\A} = (A; S^{\A}), \overline{\B} = (B; S^{\B})$. Suppose there is a homomorphism $\overline{\A}\to \overline{\B}$.
    
    If there is a homomorphism $\psi\colon \overline{\A} \to \overline{\B}$ such that
    \begin{enumerate}

    \item[(1)] $( \psi(A);\; \cdot_\psi )$ is a commutative regular monoid,

    \item[(2)] $m(x,y,z) = x\cdot_\psi y^{-1}\cdot_\psi z$ is a polymorphism of $\psi(\A)$, and

    \item[(3)] $[\psi(R^{\A})]\subseteq R^{\B}$,

    \end{enumerate}

    then $\PCSP(\overline{\A}, \overline{\B} )$ is solvable by BLP+AIP, hence is in $\mathbf{P}$. Otherwise $\PCSP(\overline{\A}, \overline{\B} )$ is $\mathbf{NP}$-hard.
    
\end{main}

\begin{proof}

First observe by Lemma \ref{lemma cosets and polymorphisms} that conditions (2) and (3) together are equivalent to $[\psi(S^\A)]\subseteq S^\B$. Therefore tractability follows from Proposition \ref{TractabilityBinaryPolOp}.

For hardness, we will assume there is no homomorphism $\psi\colon \overline{\A} \to \overline{\B}$ satisfying condition (1) of the theorem and satisfying the condition that $[\psi(S^\A)]\subseteq S^\B$. We will demonstrate \textbf{NP}-hardness of $\PCSP(\overline{\A}, \overline{\B} )$ by using techniques from \cite{larrauri2025equationsfinitemonoidsinfinite,larrauri2024solving}. For each $n\in \mathbb{N}\backslash \lbrace 0
\rbrace$, let $\overline{\A}_n$ be the diagonal substructure of $\overline{\A}^n$ with universe $A_n = \lbrace (a, \dots , a) : a\in A\rbrace$. We first remark that since (by assumption) no homomorphism $\psi$ exists satisfying condition (1) and $[\psi(S^\A)]\subseteq S^\B$, every polymorphism $p\in \Pol(\overline{\A}, \overline{\B})$ satisfies exactly one of the following:

\begin{enumerate}
    \item[(i)] $\mathbf{M}= (p(A^{\ar(p)}); \cdot_p)$ is not a commutative monoid.
    \item[(ii)] $\mathbf{M}$ is a commutative monoid, but $\mathbf{N}=(p_\Delta(A); \cdot_{p_\Delta})$ is not a regular submonoid of $\mathbf{M}$.
    \item[(iii)] $\mathbf{M}$ is a commutative monoid, $\mathbf{N}$ is a regular submonoid of $\mathbf{M}$, but $[p_\Delta(S^{\A})]\not\subseteq S^{\B}$.
\end{enumerate}

Let $\mathcal{M}_1$, $\mathcal{M}_2$, $\mathcal{M}_3$ be the sets of $p\in\Pol(\A, \B)$ satisfying conditions (i), (ii), and (iii) respectively. We will define a selection function $I_i$ on $\mathcal{M}_i$ for each $i\in [3]$.

(i) For $p\in \mathcal{M}_1$, we let $$I_1(p)= \lbrace i\in [\ar(p)]: (p_i(A) ; \cdot_p) \text{ is not a commutative monoid}\rbrace .$$ We will show $I_1$ is a selection function for $\mathcal{M}_1$ with bound $$K_1=1+2| \lbrace \text{magmas } \mathbf{C} : C\subseteq B \text{ and } \mathbf{C} \text{ is not a commutative monoid} \rbrace |.$$

To show $I_1$ is a selection function for $\mathcal{M}_1$, we will start by showing that $I_1(p)$ is nonempty for any $p\in \mathcal{M}_1$. To see this, note that $(p(A^n); \cdot_p)$ is not a commutative monoid by assumption. Then by Lemma \ref{claim 2} we know there exists $i\in [n]$ such that $( p_i(A) ; \cdot_p)$ is not a commutative monoid. Hence $I_1(p)$ is nonempty.

Now suppose $p\in \mathcal{M}_1^{(n)}$ and $\pi\colon [n]\to [m]$. Let $q=p^{(\pi)}$ be such that $q\in \mathcal{M}_1$, and let $i\in I_1(q)$. Because $q$ is a minor of $p$, we know from Lemma \ref{lemma2} that for all $a\in A$ we have $$q_i(a) = \prod_{j\in \pi^{-1}(i)}p_j(a) \in  \prod_{j\in \pi^{-1}(i)} p_j(A)$$ where the product is with multiplication $\cdot_{p}$. Hence $q_i(A)\subseteq \prod_{j\in \pi^{-1}(i)} p_j(A)$. As in the argument above, since $(q_i(A) ;\cdot_{q} )$ is not a commutative monoid, we have some $j\in\pi^{-1}(i)$ such that $(p_j(A);,\cdot_p)$ is not a commutative monoid either. Then $j\in I_1(p)$ and $I_1(q)= I_1(p^{(\pi)})\subseteq \pi(I_1(p))$. Together with the nonemptiness of $I_1(q)$, this shows that $I_1(p^{(\pi)})\cap \pi(I_1(p))\neq \emptyset$.

Next, we show that $I_1$ is bounded. Let $p\in \mathcal{M}_1^{(n)}$, let $1\leq i<j<k\leq n$, and let $i\in I_1(p)$. Then we cannot have $p_i(A) = p_j(A) = p_k(A)$ by Lemma \ref{claim 2}, because $(p_i(A); \cdot_p)$ is not a commutative monoid. The number of distinct images of $p_i(A)$ with $i\in I_1(p)$ is bounded by $$|\lbrace \text{magmas } \mathbf{C} : C\subseteq B \text{ and } \mathbf{C} \text{ is not a commutative monoid} \rbrace|.$$ Together with the previous statement, the Pigeonhole Principle implies $I_1(p)<K$.

(ii) Consider $ \mathcal{M}_2.$ For a homomorphism $\psi\colon \A\to \B$, let $\mathcal{M}_\psi = \lbrace p\in \mathcal{M}_2 : p_\Delta = \psi \rbrace$. Then note that if $q$ is a minor of $p\in \mathcal{M}_\psi$, we have $q_\Delta = p_\Delta = \psi$. So $\mathcal{M}_\psi$ is a minion and $\mathcal{M}_2$ is a finite disjoint union of such $\mathcal{M}_\psi$ because every polymorphism induces a homomorphism $\psi$ as described above. If we obtain a selection function $I_\psi$ for each such $\mathcal{M}_\psi$, then by taking their union, we get a selection function for $\mathcal{M}_2$ and we are done.

Consider a $\mathcal{M}_\psi$ with $\psi\in \mathcal{M}_2^{(1)}$ as above. Because $\mathbf{N} = (\psi(A); \cdot_\psi)$ is a commutative monoid which is not regular, we have $a\in A$ such that $\psi(a)$ is not in a subgroup of $\mathbf{N}$. Consider the monoidal minion $\mathcal{M}_{\mathbf{N}, \psi(a)}$. We claim that $$\xi\colon \mathcal{M}_\psi\to \mathcal{M}_{\mathbf{N}, \psi(a)}, \quad p \mapsto (p(a, e,\dots , e), \dots ,p(e, \dots ,e , a))$$ is a minion homomorphism. By Lemma \ref{lemma2} $$p(a, e,\dots , e) \cdot_{\psi} \dots \cdot_{\psi} p(e, \dots ,e , a) = p(a, \dots , a) = \psi(a).$$ Hence $\xi(p) \in \mathcal{M}_{\mathbf{N},\psi(a)}$ Moreover, $\xi$ preserves minors. To see this, let $p \in \mathcal{M}_2^{(n)}$, $\pi\colon [n]\to [m]$, and let $q = p^{(\pi)}$. Then 
\begin{align*}
    \xi(p^{(\pi)}) = \xi(q) &= (q_1(a), \dots ,q_m(a))\\ 
    &=\left(  \prod_{j\in \pi^{-1}(1)} p_j( a), \dots, \prod_{j\in \pi^{-1}(m)} p_j( a)\right) = \xi(p)^{(\pi)}.
\end{align*}


Now Proposition \ref{SelectionCorollary} and Proposition \ref{HomomorphismToUnionOfMonoidalMinions} provide us with a selection function $I_\psi$ for $\mathcal{M}_\psi$ with some finite bound $K_\psi$. We take the union of the finitely many $\mathcal{M}_\psi$ for $\psi \in \mathcal{M}_2$ to get a selection function for $\mathcal{M}_2$ with bound $K_2= \max\lbrace K_\psi : \psi \in \mathcal{M}_2^{(1)} \rbrace$. This concludes the proof that there is a selection function $I_2$ for $\mathcal{M}_2$.

(iii) Let $p\in \mathcal{M}_3$. Note for this case, we have that $p$ satisfies conditions $(1)$ and $(2)$ of Lemma \ref{biglemma}. So define $$I_3\colon{\mathcal{M}_3}\to \mathcal{P}(\mathbb{N}), \quad p\mapsto \lbrace i \in [\ar(p)]: i \text{ is in a } p\text{-equivalence class of size smaller than } L \rbrace$$ where $L$ is from Lemma \ref{biglemma}. Since $[p_\Delta(S^{\A})]\not\subseteq S^{\B}$ and $p((S^{\A})^{\ar(p)})\subseteq S^{\B}$, we cannot have $[p_\Delta(S^{\A})] \subseteq p((S^{\A})^{\ar(p)})$. Hence by applying Lemma \ref{biglemma}, we obtain $\pi(I_3(p))\cap I_3(p^{(\pi)}) \neq \emptyset$.

It remains to show that $I_3$ is bounded. For all $p\in \mathcal{M}_3$, every $p$-equivalence class is determined by functions $p_i\colon A\to B$. Since there are at most $|B|^{|A|}$ many such functions, there are at most $|B|^{|A|}$ many $p$-equivalence classes for each $p\in \mathcal{M}_3$. But by definition of $I_3(p)$, the size of each equivalence class contained in $I_3(p)$ is less than $L$. Then $|I_3(p)| \leq |B|^{|A|}L$ for all $p\in \mathcal{M}_3$. Since we have found a selection function for $\mathcal{M}_3$, this concludes the proof.\end{proof}

For algebras $\A$ and $\B$ without the additional relations $R^{\A}$ and $R^{\B}$, we may omit condition (3) from Theorem \ref{DichotomyBinaryTermOp} and get the following result: 

\begin{main}
    \label{MainDichotomy}
    Let $\A=(A; f_1^{\A}, \dots , f_r^{\A})$ and $\B= (B; f_1^{\B}, \dots , f_r^{\B})$ be finite algebras, let $e\in A$, let $\cdot$ be a binary polynomial operation on $\A$ such that $x \cdot e = x= e \cdot x$ for all $x\in A$. Suppose there is a homomorphism $\A\to \B$.
    
    If there is a homomorphism $\psi\colon \A \to \B$ such that
    \begin{enumerate}

    \item[(1)] $( \psi(A);\; \cdot_\psi )$ is a commutative regular monoid, and

    \item[(2)] $m(x,y,z) = x\cdot_\psi y^{-1} \cdot_\psi z$ is a polymorphism of $\psi(\A)$,

    \end{enumerate}

    then $\PEqn({\A}, {\B} )$ is solvable by BLP+AIP, hence is in $\mathbf{P}$. Otherwise $\PEqn({\A}, {\B} )$ is $\mathbf{NP}$-hard.
\end{main}

\begin{corollary}
    \label{LatticeCorollary}
    Let $\A$ and $\B$ be expansions of finite lattices with a homomorphism $\A \to \B$. If there is a homomorphism $\psi\colon \A \to \B$ such that $|\psi(A)| = 1$, then $\PEqn(\A, \B)$ is in $\mathbf{P}$. Otherwise $\PEqn(\A , \B)$ is $\mathbf{NP}$-hard.
\end{corollary}

\begin{proof}
    The binary operation $\wedge$ of a lattice has two sided identity $1$. Moreover, every element of a lattice belongs to a trivial subgroup (and is its own inverse) with respect to the operation $\wedge$. Then we write $m(x,y,z) = x\wedge y \wedge z$. Note that if $m$ commutes with $\vee$, then in particular $$0 = (1 \wedge 1 \wedge 0) \vee (0 \wedge 0 \wedge 1)= (1 \vee 0) \wedge (1 \vee 0) \wedge (0 \vee 1) = 1.$$ We apply Theorem \ref{MainDichotomy} using the operation $\wedge$ to obtain the result.
\end{proof}

We observe that Mal'cev algebras also fall into the hypotheses of Theorem \ref{MainDichotomy}. Consider an algebra $\A$ with Mal'cev polynomial operation $m$. We see that for $e \in A$, the polynomial operation $x\cdot y =m(x,e,y)$ satisfies $x\cdot e = x = e\cdot x$. However, we are able to simplify the statement of the theorem for algebras with a Mal'cev polynomial using the following lemma:

\begin{lemma}
    \label{AbelianMal'cev}
    Let $\A$ be a finite algebra with a Mal'cev polynomial operation $m$, let $e\in A$, and define $x\cdot y = m(x,e,y)$ for all $x,y\in A$. Suppose that $(A; \cdot)$ is a commutative, regular monoid.

    Then $\A$ is abelian if and only if $t(x,y,z) = x\cdot y^{-1}\cdot z$ is a polymorphism of $\A$.
\end{lemma}

\begin{proof}
    Assume $\A$ is abelian and let $\overline{\A}$ be the expansion of $\A$ with constant symbols for each $a\in A$. Then $\overline{\A}$ is abelian and $m$ is a Mal'cev term operation of $\overline{\A}$. Note that $m(x,y,z) = x-y+z$ for some abelian group operation $+$ on $\overline{\A}$ by \cite[Proposition 5.7]{CommutatorTheory}. We have $x\cdot y = m(x,e,y) = x - e +y$. Therefore $y\cdot (-y+2e) = y - e +(- y + 2e) = e$. That is, the element $-y +2e$ is an inverse of the element $y$ with respect to the operation $\cdot$. So $(A; \cdot)$ forms an abelian group with identity $e$ and with $y^{-1} = -y+2e$ for each $y\in A$. But then $$t(x,y,z) = x\cdot y^{-1}\cdot z = x-e + (- y +2e) - e +z = x- y + z = m(x,y,z).$$ Since $m$ is a polymorphism of $\overline{\A}$, we see that $t$ is a polymorphism of $\mathbf{\overline{A}}$. Hence $t$ is also a polymorphism of $\A$.
    
    For the other direction, assume that $t$ is a polymorphism of $\A$. Because $t$ is idempotent and $m$ is a polynomial operation of $\A$, we know that $t$ commutes with $m$. Then for $x,y,z\in A$, we have
\begin{align*}
        m(x,y,z) & = m(t(x,e,e), t(y,e,e) , t(e,e,z))\\
        & = t(m(x,y,e),m(e,e,e), m(e,e,z) )\\
        & = t(m(x,y,e), e, z)\\
        & = m(x,y,e)\cdot e^{-1}\cdot z\\
        & = m(x,y,e)\cdot z.
    \end{align*}
Similarly, we have $m(x,y,z) = x\cdot m(e,y,z)$. Therefore $e = m(y,y,e) = y\cdot m(e,y,e)$. Hence $(A; \cdot)$ is actually an abelian group with inverses given by $y^{-1} = m(e,y,e)$. Now we have
\begin{align*}
        m(x,y,z) &= x \cdot m(e,y,e)\cdot z\\
        &=x\cdot y^{-1}\cdot z\\
        &=t(x,y,z).
    \end{align*}
Since $t$ is a polymorphism of $\A$ by assumption, we have that $m$ is a polymorphism of $\mathbf{\overline{A}}$, and $\mathbf{\overline{A}}$ is abelian by \cite[Proposition 5.7]{CommutatorTheory}. Therefore $\A$ is abelian as well.
\end{proof}

The following lemma is folklore.

\begin{lemma}
    \label{Mal'cevPolynomial}
    Let $\A$ be an algebra with a Mal'cev polynomial operation $p$. If $\A$ is abelian, then $p$ is a term operation of $\A$.
\end{lemma}

\begin{proof}
    Let $\overline{\A}$ be the expansion of $\A$ with $0$-ary operations for each $a\in A$ interpreted in the natural way. Then $\overline{\A}$ is abelian and $p$ is a Mal'cev term operation of $\overline{\A}$. And $p(x,y,z) = x-y+z$ for all $x,y,z\in A$ for some abelian group operation $+$ on $A$ by \cite[Lemma 5.6, Proposition 5.7]{CommutatorTheory}. Because $p$ is a polynomial operation on $\A$, there exists some $n\in \mathbb{N}$, some term operation $t$ on $\A$, and some $a_1, \dots ,a_n \in A$ with $$p(x,y,z) = t(x,y,z,a_1, \dots ,a_n)$$ for all $x,y,z\in A$. Since $\A$ is polynomially equivalent to a  module by \cite[Corollary 5.9]{CommutatorTheory}, there exists some polynomial $c$ over $\A$ with $$t(x,y,z, w_1, \dots , w_n) = x-y+z+c(w_1, \dots , w_n)$$ for all $x,y,z,w_1, \dots ,w_n\in A$ and with $c(a_1, \dots , a_n) = 0$.

    It is now straightforward to check that $$t(x, t(y, \dots y), z , y, \dots , y) = p(x,y,z)$$ for all $x,y,z\in A$. \end{proof}

We can now state a dichotomy for Mal'cev algebras.

\begin{main}
    \label{DichotomyMal'cev}
    Let $\A$, $ \B$ be finite algebras with a homomorphism $\A\to \B$ such that $\A$ has a Mal'cev polynomial operation $m$.
    
    If there is a homomorphism $\psi\colon\A\to \B$ such that $\psi(\A)$ is abelian, then $\PEqn(\A, \B)$ is solvable by $\AIP$, hence is in $\mathbf{P}$. Otherwise $\PEqn(\A, \B)$ is $\mathbf{NP}$-hard.
\end{main}

\begin{proof}
    For tractability, we note that $\psi(\A)$ has a Mal'cev term by Lemma \ref{Mal'cevPolynomial}. Then we can construct alternating polymorphisms from $\A$ to $\B$ of every odd arity as in Proposition \ref{TractabilityTaylorAlg}.

    For hardness, assume there is no homomorphism $\psi\colon \A\to \B$ such that $\psi(\A)$ is abelian. Let $e\in A$. Then conditions (1) and (2) of Theorem \ref{MainDichotomy} cannot hold for $x\cdot y = m(x,e,y)$, since otherwise $\psi(\A)$ would be abelian by Lemma \ref{AbelianMal'cev}. Hence $\PEqn(\A, \B)$ is $\mathbf{NP}$-hard.
\end{proof}

\begin{corollary}
    Let $\A$ and $\B$ be finite rings with a homomorphism $\A \to \B$.

    If there is a homomorphism $\psi\colon \A \to \B$ such that $\psi(\A)$ has a $0$-multiplication, then $\PEqn(\A , \B)$ is solvable by $\AIP$, hence is in $\mathbf{P}$. Otherwise $\PEqn(\A , \B)$ is $\mathbf{NP}$-hard.
\end{corollary}

\begin{proof}
    This follows from Theorem \ref{DichotomyMal'cev} since a ring is abelian if and only if it has a $0$-multiplication by Example \ref{AbelianRing}.
\end{proof}

We can now return to Example \ref{ExampleRings}. Here, we have $\A = (2\mathbb{Z}_8; +^{\mathbb{Z}_8},\cdot^{\mathbb{Z}_8} , 2^{\mathbb{Z}_8})$ and $\B = (\mathbb{Z}_4; +^{\mathbb{Z}_4},\cdot^{{\mathbb{Z}_4}}, 2^{\mathbb{Z}_4})$. The only homomorphism $\A \to \A$ is the identity, and $\A$ does not have a $0$-multiplication, so $\SysTerm(\A) = \PEqn(\A, \A)$ is $\mathbf{NP}$-complete. Similarly, $\SysTerm(\B)$ is $\mathbf{NP}$-complete. There is only one homomorphism $\psi\colon\A \to \B$ determined by $2^{\mathbb{Z}_8}\mapsto 2^{\mathbb{Z}_4}$. And $\psi(\A)$ is a two element ring with a $0$-multiplication. Hence $\PEqn(\A , \B)$ is in $\mathbf{P}$.

\section{Congruence Modular Varieties}\label{Congruence Modular Varieties}

We will now extend Theorem \ref{DichotomyMal'cev} to algebras in congruence modular varieties. For two congruences $\alpha$, $\beta$ on an algebra $\A$, recall that we denote the commutator of $\alpha$ with $\beta$ by $[\alpha , \beta]$.

By \cite[Theorem 6.2]{CommutatorTheory}, we know that solvable algebras in congruence modular varieties are Mal'cev algebras. With Theorem \ref{DichotomyMal'cev}, we have already established a dichotomy for $\PEqn(\A , \B)$ in the case where $\A$ is a Mal'cev algebra.  In the remainder of this section, we will establish a dichotomy for $\PEqn(\A , \B)$ in the case where $\A$ is in a congruence modular variety by using the previous results and by obtaining a selection function for $$ \lbrace p\in \Pol(\A , \B) : p(\A^{\ar(p)}) \text{ is not solvable}\rbrace .$$

\begin{lemma}
    \label{SelectionFunctionNotSolvable}
    Let $\A$ and $\B$ be finite algebras with a homomorphism $\A \to \B$ such that $\A$ is in a congruence modular variety. For each $n\in \mathbb{N}$ and each $i\in [n]$, let $\theta^{(n)}_i$ be the kernel of the $i$-th projection map $\A^n\to \A$. For $$\mathcal{M}=\lbrace p\in \Pol(\A , \B) : p(\A^{\ar(p)}) \text{ is not solvable}  \rbrace$$ we define $$I\colon \mathcal{M}\to \mathcal{P}(\mathbb{N}), \quad p\mapsto \lbrace i\in [n]: \A^n/(\theta^{(n)}_i \vee \ker p ) \text{ is not solvable}\rbrace .$$ Then $I$ is a selection function with bound $\max \lbrace|\Con \;\C| : \C \leq \B \rbrace $.
\end{lemma}

\begin{proof}
We first claim that \begin{equation}
\label{EqNonempty}
    I(p)\neq \emptyset \text{ for all } p\in \mathcal{M}
\end{equation} To prove this claim, let $p\in \Pol(\A, \B)^{(n)}$ such that $\A^n/(\theta_i^{(n)}\vee \ker p)$ is solvable for each $i\in [n]$. Then we have some $d\geq 0$ such that $[1]^d\leq \theta_i^{(n)}\vee \ker p$ for each $i\in [n]$ by the remarks following Proposition \ref{CommMeetJoin}. We obtain with the left associated iterated binary commutator:
\begin{align*}
    [1]^{d+n-1} &= [[1]^d,[1]^d,[1^{d+1}], \dots ,[1]^{d+n-2}]\\
    & \leq \underbrace{[[1]^d, \dots , [1]^d]}_{n} \quad \text{by Proposition }\ref{CommMonotone}\\
    &\leq [\theta_{1}^{(n)}\vee \ker p, \dots , \theta_{n}^{(n)}\vee \ker p]\\
    &\leq [\theta_{1}^{(n)}, \dots , \theta_{n}^{(n)}] \vee \ker p \quad\text{by Proposition }\ref{CommMeetJoin}(2)\\
    &\leq \bigwedge_{i\in [n]}\theta_i^{(n)}\vee \ker p \quad \text{by Proposition }\ref{CommMonotone}.
\end{align*}
Therefore $$\A^n/ (\bigwedge_{i\in [n]}\theta_i^{(n)} \vee \ker p) = \A / \ker p$$ is solvable, hence $p\not\in \mathcal{M}$ and \ref{EqNonempty} is proved.

Next, let $p\in \mathcal{M}^{(n)}$ and $\pi\colon[n]\to [m]$ such that $p^{(\pi)}\in \mathcal{M}$. Then we claim \begin{equation}\label{EqSubset}\pi(I(p)) \subseteq I(p^{(\pi)}). \end{equation} In particular, by \ref{EqNonempty}, this will yield $\pi(I(p))\cap I(p^{(\pi)}) \neq \emptyset$.

For the proof, we will introduce some new notation. For $x\in A^m$, we define $x_\pi\in A^n$ by $(x_\pi)_i = x_{\pi(i)}$ for each $i\in [n]$. Note that $p(x_\pi)= p^{(\pi)}(x)$. Now for $i\in [n]$, define $$h\colon\A^m\to \A^n/(\theta_i^{{(n)}}\vee \ker p), \quad x\mapsto x_\pi/(\theta_i^{(n)}\vee \ker p).$$ We claim that $h$ is a surjective homomorphism whose kernel contains $\theta_{\pi(i)}^{(m)} \vee \ker p^{(\pi)}$. By definition, $h$ is a homomorphism. To see $h$ is surjective, note that if $y\in A^n$, then there exists some $x\in A^m$ with $(x_\pi)_i = y_i$. Now note that $\theta_{\pi(i)}^{(m)}\leq \ker h$ because for $x,y\in A^m $ we know $$x\; \theta_{\pi(i)}^{(m)} \; y \Leftrightarrow x_\pi \;\theta_i^{(n)} \;y_\pi \Rightarrow x \;\ker h \; y.$$ We also have $\ker p^{(\pi)}\leq \ker h$ because $$x \; \ker p^{(\pi)} \;y \Leftrightarrow x_\pi \ker p \;y_\pi \Rightarrow x \;\ker h \; y .$$ Therefore $\theta_{\pi(i)}^{(m)} \vee \ker p^{(\pi)} \leq \ker h$.

We now know that $h(\A^m) =\A^n/(\theta_i^{{(n)}}\vee \ker p)$ is a homomorphic image of $\A^m/(\theta_{\pi(i)}^{(m)} \vee \ker p^{(\pi)})$. By the remarks following Proposition \ref{CommMeetJoin}, we know that if $\A^m/(\theta_{\pi(i)}^{(m)} \vee \ker p^{(\pi)})$ is solvable, then $\A^n/(\theta_i^{{(n)}}\vee \ker p)$ is solvable. Equivalently, if $\A^n/(\theta_i^{{(n)}}\vee \ker p)$ is not solvable, then $\A^m/(\theta_{\pi(i)}^{(m)} \vee \ker p^{(\pi)})$ is not solvable, which is to say if $i\in I(p)$, then $\pi(i)\in I(p^{(\pi)})$. Hence \ref{EqSubset} is proved, and together with \ref{EqNonempty}, we have $\pi(I(p))\cap I(p^{(\pi)})\neq \emptyset$.

Now we claim that for all $p\in \mathcal{M}$, we have 

\begin{equation}\label{EqBound}|I(p)|\leq \max \lbrace|\Con\; \C|: \C \leq \B \rbrace .\end{equation} To prove this claim, let $p\in \mathcal{M}^{(n)}$, and let $i,j\in [n]$ with $i\neq j$. Note that $\theta_i^{(n)} \vee \theta_j^{(n)} = 1_{\A^n}$. So if $\theta_i^{(n)} \vee \ker p = \theta_j^{(n)} \vee \ker p$, we have $\theta_i^{(n)} \vee \ker p = 1_{\A^n}$. Hence $i\not\in I(p)$. This proves \ref{EqBound}. Together with the fact that $\pi(I(p))\cap I(p^{(\pi)})\neq \emptyset$, we have shown that $I$ is a selection function on $\mathcal{M}$. \end{proof}

We are now ready to give a dichotomy for $\PEqn(\A, \B)$ in the case where $\A$ is in a congruence modular variety.

\begin{main}
    \label{DichotomyCongruenceModular}
    Let $\A$ and $\B$ be finite algebras with a homomorphism $\A\to \B$ such that $\A$ is in a congruence modular variety.

    If there is a homomorphism $\psi\colon \A \to \B$ such that $\psi(\A)$ is abelian, the $\PEqn(\A, \B)$ is solvable by $\AIP$, hence is in $\mathbf{P}$. Otherwise $\PEqn(\A, \B)$ is $\mathbf{NP}$-hard.
\end{main}

\begin{proof}
    Tractability using AIP follows from Proposition \ref{TractabilityTaylorAlg}.
    
    For hardness, suppose there is no homomorphism $\psi\colon \A\to \B$ such that $\psi(\A)$ is abelian. Let $$\mathcal{M} = \lbrace p\in \Pol(\A, \B):p(\A^{\ar(p)}) \text{ is solvable}\rbrace.$$ From Lemma \ref{SelectionFunctionNotSolvable}, we have a selection function for $\Pol(\A, \B)\backslash \mathcal{M}$. Therefore, it suffices to find a selection function for $\mathcal{M}$ by Proposition \ref{SelectionFunctionCriteria}. We first observe that $\mathcal{M}$ is a minion as subalgebras of solvable algebras are solvable by Proposition \ref{CommMonotone}(3). Let $d\in \mathbb{N}$ be such that $\widehat{\A} = \A / [1]^d$ is the maximal solvable quotient of $\A$. Hence $\widehat{\A}^n \cong \A^n/ [1]^d \times \dots \times [1]^d$ is the maximal solvable quotient of $\A^n$ for each $n\in \mathbb{N}$ by Proposition \ref{CommDirectProd}. Note that $\widehat{\A}$ is a Mal'cev algebra by Proposition \ref{SolvableInConMod}. By assumption, there is no homomorphism $\psi\colon \widehat{\A}\to\B$ such that $\psi(\widehat{\A})$ is abelian. From the proofs of Theorem \ref{DichotomyBinaryTermOp} and Theorem \ref{DichotomyMal'cev}, there is a selection function for $\Pol(\widehat{\A}, \B)$. It therefore suffices by Proposition \ref{HomomorphismToUnionOfMonoidalMinions} to find a minion homomorphism $\xi\colon \mathcal{M}\to \Pol(\widehat{\A}, \B)$. We define $$\xi\colon \mathcal{M}\to \Pol(\widehat{\A}, \B) , \quad p\mapsto \widehat{p}$$ where $\widehat{p}$ is given by $$\widehat{p}\colon \widehat{\A}^n\to\B, \quad \widehat{a}\mapsto p(a)$$ and $n=\ar(p)$. Note $\widehat{p}$ is well defined because the assumption that $p(\A^n)$ is solvable implies $\underbrace{[1]^d\times \dots \times [1]^d}_{n} \leq \ker p$ by Proposition \ref{CommDirectProd} and the remarks following Proposition \ref{CommMeetJoin}. It is immediate from the definition of $\xi$ that $\xi$ preserves arity minors. This concludes the proof that $\xi$ is a minion homomorphism.\end{proof}

\section{Metaproblem Results}

In \cite[Theorem 2]{Mayr23}, Mayr establishes a dichotomy for $\SysTerm(\A)$ in the case where $\A$ is a finite algebra in a congruence modular variety. In particular, $\SysTerm(\A)$ is in $\mathbf{P}$ if the \emph{core} of $\A$ is abelian and $\SysTerm(\A)$ is $\mathbf{NP}$-complete otherwise. A \emph{core} of an algebra $\A$ is a minimal image of an endomorphism of $\A$. Then, in \cite[Theorem 3]{Mayr23}, Mayr gives a quasi-polynomial time algorithm which decides whether or not the core of an input algebra $\A$ is abelian:

\begin{theorem}
    \label{quasi-poly-metaproblem}
    \cite[Theorem 3]{Mayr23}
    Let $\A$ be a finite algebra in a variety with a weak difference term. There exists a quasi-polynomial time algorithm which decides whether or not the core of $\A$ is abelian. If the core of $\A$ is abelian, then the core of $\A$ and the graph of its Mal'cev term operation can also be computed in quasi-polynomial time.
\end{theorem}

This decides the metaproblem for $\SysTerm(\A)$ in quasi-polynomial time. Theorem \ref{DichotomyCongruenceModular} is a generalization of \cite[Theorem 2]{Mayr23}. The complexity of $\PEqn(\A , \B)$ in the case where $\A$ is in a congruence modular variety is determined by whether or not there is a homomorphism from $\A $ to $\B$ with an abelian image. To settle the metaproblem for $\PEqn(\A , \B)$ in quasi-polynomial time, we can modify the algorithm from \cite[Theorem 3]{Mayr23}, which determines whether or not there is a homomorphism from $\A$ to itself with abelian image. To obtain the quasi-polynomial time algorithm, we first need the following lemma.

\begin{lemma}
    \label{MaxAbelianQuotient}
    Let $\A$ be a finite algebra in a variety with a weak difference term $d$ and let $\B$ be finite algebra. Then there exists a homomorphism $\psi\colon \A\to \B$ such that $\psi(\A)$ is abelian if and only if there exists a homomorphism $\A/[1_\A, 1_\A]\to \B$.
\end{lemma}

\begin{proof}
Let $\varphi \colon \A\to \A/[1_\A, 1_\A], \quad a\mapsto a/[1_\A, 1_\A]$ be the canonical homomorphism, and suppose $\widehat\psi\colon \A/[1_\A, 1_\A]\to \B$ is a homomorphism. Then we claim $\widehat{\psi}\circ \varphi\colon\A\to \B$ has abelian image. Recall that $\widehat{\A} = \A/[1_\A, 1_\A]$ is abelian. By Definition \ref{WeakDifference} and Proposition \ref{CommMonotone}, the term $d$ satisfies $$d^\A(x,x,y)\; [1_\A, 1_\A]\; y \;[1_\A, 1_\A]\;d^\A(y,x,x)$$ for all $x,y\in A$. Equivalently $$d^{\widehat{\A}}(x,x,y)= y =d^{\widehat{\A}}(y,x,x)$$ for all $x,y\in \widehat{A}$. Hence $\widehat{\A}$ is an abelian Mal'cev algebra. In particular $\widehat{\A}$ is in a congruence modular variety. And homomorphic images of abelian algebras in congruence modular varieties are abelian by the comments following Proposition \ref{CommMeetJoin}.

Conversely, suppose $\psi \colon \A\to \B$ is such that $\psi(\A)$ is abelian. Since $\psi(\A) \cong \textbf{A}/\ker\psi$ is abelian, we know that $\ker(\psi)\geq [1,1]$ and that $\A/\ker\psi \cong (\A/[1,1])/ (\ker \psi / [1,1])$. Now the map $$\widehat{\psi}\colon \A/[1,1]\to \psi(\A), \quad \widehat{a}\mapsto \psi(a)$$   is seen to be the desired well defined homomorphism with kernel $\ker \psi / [1,1]$.
\end{proof}

The algorithm in \cite[Theorem 3]{Mayr23} decides in quasi-polynomial time whether or not there is a homomorphism from $\widehat{\A}$ to $\A$. Similarly, we have:

\begin{theorem}
\label{QuasiPolyAlg}
    Let $\A$ be a finite algebra in a variety with a weak difference term and let $\B$ be a finite algebra. There exists a quasi-polynomial time algorithm which decides whether or not there exists a homomorphism from $\widehat{\A} =\A/[1,1]$ to $\B$.
\end{theorem}

\begin{proof}
    We will modify the algorithm from \cite[Theorem 3]{Mayr23}, which checks if there is a homomorphism from $\widehat{\A}$ to $\A$ in quasi-polynomial time. We instead check if there is a homomorphism from $\widehat{\A}$ to $\B$.

    As in Lemma \ref{MaxAbelianQuotient}, we know that $\widehat{\A}$ is an abelian algebra with a Mal'cev term operation $m^{\widehat{\A}}$. And $m^{\widehat{\A}}(x,y,z) = x-y+z$ where $+$ is an abelian group operation on $\widehat{\A}$ by \cite[Lemma 5.6, Proposition 5.7]{CommutatorTheory}. The algebra $(\widehat{A};m^{\widehat{\A}})$ can be computed in polynomial time as in \cite[Theorem 3]{Mayr23}. We may compute a set of generators $\lbrace u_1, \dots ,u_n \rbrace$ for $(\widehat{A};m^{\widehat{\A}})$ in time polynomial in $\widehat{\A}$ with $n\leq \log |A|+1$. Every homomorphism $\widehat{\A}\to \B$ is uniquely determined by its image on $\lbrace u_1, \dots ,u_n \rbrace$, so we may choose any  $v_1, \dots , v_n \in B$ to be the possible images of $u_1, \dots , u_n$. Then we compute $\langle (u_1,v_1), \dots ,(u_n,v_n)\rangle \leq \widehat{\A}\times \B$ in time polynomial in $\widehat{\A}$ and $\B$. By computing this subalgebra, we can test whether or not the map determined by $u_i \mapsto v_i$ for $i\in [n]$ extends to a homomorphism. There are $|B|^n \leq 2^{\log |B|(\log |A|+1)}$ choices for $v_1, \dots , v_n$. Hence this algorithm runs in quasi-polynomial time in $\A$ and $\B$.
\end{proof}

We can now prove Theorem \ref{MetaProblemMal'cev}: \begin{proof}[Proof of Theorem \ref{MetaProblemMal'cev}]
    By combining Theorem \ref{DichotomyMal'cev} with Lemma \ref{MaxAbelianQuotient} and Proposition \ref{QuasiPolyAlg}, we obtain the result.
\end{proof}

\section*{Aknowledgment}Thanks to Peter Mayr for his suggestions on the approach for Section \ref{Congruence Modular Varieties}.

\section*{Statements on use of AI}
AI has not been used in the writing of this paper in any way.

\bibliographystyle{plain}
\bibliography{bibliography}

\end{document}